\documentclass[11pt]{article}

\usepackage{amsmath,amssymb,amsthm,mathtools}
\usepackage[margin=1in]{geometry}
\usepackage{hyperref}
\usepackage[nameinlink,capitalize]{cleveref}
\usepackage{xspace}

\hypersetup{colorlinks=true,linkcolor=blue,citecolor=blue,urlcolor=black,linkbordercolor={0 0 1}}
\usepackage{tikz-cd}

\newcount\Comments
\newcommand{\nc}{\newcommand}

\usepackage[nameinlink,capitalize]{cleveref}
\crefformat{equation}{#2(#1)#3}
\Crefformat{equation}{#2(#1)#3}

\Crefformat{figure}{#2Figure #1#3}
\Crefname{assumption}{Assumption}{Assumptions}
\Crefformat{assumption}{#2Assumption #1#3}
   \Crefname{question}{Question}{Questions}
   \Crefformat{question}{#2Question #1#3}
   \Crefname{claim}{Claim}{Claims}
   \Crefformat{claim}{#2Claim #1#3}
   \Crefname{problem}{Problem}{Problems}
  \Crefformat{problem}{#2Problem #1#3}
\Crefname{subsubsection}{Section}{Sections}
\crefformat{subsubsection}{#2Section #1#3}
\Crefformat{subsubsection}{#2Section #1#3}

\nc{\sups}[1]{^{\scriptscriptstyle{#1}}}
\nc{\subs}[1]{_{\scriptscriptstyle{#1}}}

\newcommand{\wb}{\widebar}

\nc{\Critic}{\texttt{Critic}\xspace}
\nc{\PSDPUCB}{\texttt{PSDP-UCB}\xspace}
\nc{\LSVIUCB}{\texttt{LSVI-UCB}\xspace}
\nc{\Actor}{\texttt{Actor}\xspace}
\nc{\EstFeature}{\texttt{EstFeature}\xspace}
\nc{\ExpFTPL}{\texttt{ExpFTPL}\xspace}
\nc{\dist}{\mathrm{dist}}
\nc{\Bquad}{B^{\mathsf{quad}}}

\newcommand{\F}{\mathbb{F}}

\makeatletter
\newtheorem*{rep@theorem}{\rep@title}
\newcommand{\newreptheorem}[2]{%
\newenvironment{rep#1}[1]{%
 \def\rep@title{#2 \ref{##1}}%
 \begin{rep@theorem}}%
 {\end{rep@theorem}}}
\makeatother

\makeatletter
\newcommand\xlabel[2][]{\phantomsection\def\@currentlabelname{#1}\label{#2}}
\makeatother

\theoremstyle{plain}
\newtheorem{theorem}{Theorem}
\newtheorem{lemma}[theorem]{Lemma}

\newtheorem{conjecture}[theorem]{Conjecture}

\theoremstyle{definition}
\newtheorem{definition}{Definition}

\numberwithin{theorem}{section}
\numberwithin{definition}{section}

\AddToHook{env/theorem/begin}{\crefalias{theorem}{theorem}}
\AddToHook{env/lemma/begin}{\crefalias{theorem}{lemma}}
\AddToHook{env/corollary/begin}{\crefalias{theorem}{corollary}}
\AddToHook{env/conjecture/begin}{\crefalias{theorem}{conjecture}}
\AddToHook{env/proposition/begin}{\crefalias{theorem}{proposition}}
\AddToHook{env/fact/begin}{\crefalias{theorem}{fact}}
\AddToHook{env/postulate/begin}{\crefalias{theorem}{postulate}}
\AddToHook{env/claim/begin}{\crefalias{theorem}{claim}}
\AddToHook{env/assumption/begin}{\crefalias{theorem}{assumption}}
\AddToHook{env/definition/begin}{\crefalias{definition}{definition}}
\AddToHook{env/defn/begin}{\crefalias{definition}{defn}}
\AddToHook{env/example/begin}{\crefalias{definition}{example}}
\AddToHook{env/remark/begin}{\crefalias{definition}{remark}}
\AddToHook{env/question/begin}{\crefalias{definition}{question}}
\AddToHook{env/problem/begin}{\crefalias{definition}{problem}}

\nc{\DMO}{\DeclareMathOperator}

\DeclareMathOperator*{\argmax}{arg\,max}

\DMO{\prox}{prox}
\DMO{\UCB}{UCB}
\DMO{\LCB}{LCB}
\nc{\phidiff}{\phi\sups{\Delta}}
\nc{\pexp}{q_{\mathrm{exp}}}
\nc{\nn}{\nonumber}
\nc{\rk}{\mathrm{rk}}
\nc{\brk}[3]{{\rm br}_{#1}^{#2}({#3})}
\nc{\co}{{\rm co}}
\nc{\br}[2]{{\rm br}^{#1}({#2})}
\nc{\tA}{\textsc{A}}
\nc{\child}[2]{{\rm ch}_{#1}({#2})}
\nc{\parent}{\mathsf{pa}}
\nc{\dg}{\dagger}
\nc{\bB}{\mathbf{B}}
\nc{\Span}{\mathsf{span}}
\nc{\unif}{\mathsf{unif}}
\nc{\indsig}[2]{\mathcal{I}_{#1}({#2})}
\nc{\early}{{\rm pre}}
\nc{\zsink}{z_{\rm sink}}
\nc{\lowv}{{\rm low}}
\nc{\ol}{\overline}
\nc{\ul}{\underline}
\nc{\madec}[3]{\texttt{ma-dec}_{#1}({#2}, {#3})}
\nc{\madeco}[1]{\texttt{ma-dec}_{#1}}
\nc{\madecd}[3]{\texttt{ma-dec}^{\texttt{d}}_{#1}({#2}, {#3})}
\nc{\SF}{\mathscr{F}}
\nc{\SH}{\mathscr{H}}
\nc{\SP}{\mathscr{P}}
\nc{\SPc}{\wb{\mathscr{P}}}
\nc{\SB}{\mathscr{B}}
\nc{\SC}{\mathscr{C}}
\nc{\BS}{\mathbb{S}}
\nc{\PiMarkov}{\Pi^{\rm markov}}
\nc{\trunc}[2]{\mathsf{trunc}_{#2}({#1})}
\nc{\sbl}{of strong Bellman type\xspace}
\nc{\inormal}[1][\Phi, u,v]{\til{N}_{{#1}}}

\nc{\gamvec}{\gamma}
\nc{\til}{\widetilde}
\nc{\td}{\tilde}
\nc{\wh}{\widehat}
\nc{\old}[1]{\ifnum\Comments=1 {\color{brown}  [OLD: #1]}\fi}
\nc{\noah}[1]{\ifnum\Comments=1 {\color{purple} [ng: #1]}\fi}
\nc{\dhruv}[1]{\ifnum\Comments=1 {\color{red} [dr: #1]}\fi}
\nc{\mir}[1]{\ifnum\Comments=1 {\color{teal} [mc: #1]}\fi}
\nc{\sam}[1]{\ifnum\Comments=1 {\color{red} [sg: #1]}\fi}
\nc{\BP}{\mathbb{P}}
\nc{\BI}{\mathbb{I}}
\nc{\midpoint}[1][\Phi,\phi_1,\phi_2]{\mu^{\star}_{{#1}}}

\nc{\fools}[3]{\MF_{#3}({#1}, {#2})}
\nc{\fool}[2]{\MF({#1},{#2})}
\nc{\clip}[2]{{\rm clip}\left[ \left. {#1} \right| {#2} \right]}
\nc{\imax}{\omega}
\DMO{\conv}{conv}
\nc{\MH}{\mathcal{H}}
\nc{\MV}{\mathcal{V}}
\nc{\MC}{\mathcal{C}}
\nc{\MI}{\mathcal{I}}
\nc{\st}{\star}
\nc{\lng}{\langle}
\nc{\rng}{\rangle}
\DMO{\OOPT}{opt}
\nc{\dopt}[2]{\ell_{\OOPT}({#1},{#2})}
\nc{\MG}{\mathcal{G}}
\nc{\MP}{\mathcal{P}}
\nc{\PP}{\mathbb{P}}
\nc{\TT}{\mathbb{T}}
\nc{\TTmax}{\TT_{\max}}
\DMO{\REG}{Reg}
\DMO{\WREG}{wReg}
\nc{\reg}[2]{{\Delta}_{{#1}}({#2})}
\nc{\wreg}[2]{{\Delta}^{\rm w}_{{#1}}({#2})}
\nc{\Reg}[2]{{\REG}_{{#1}}({#2})}
\nc{\wReg}[2]{{\WREG}_{{#1}}({#2})}
\DMO{\Gap}{Gap}
\DMO{\GD}{GD}
\DMO{\GDA}{GDA}
\DMO{\EG}{EG}
\nc{\TE}{\til{\E}}
\nc{\Var}{\mathbf{Var}}
\DMO{\Cov}{Cov}
\DMO{\OGDA}{OGDA}
\DMO{\Unif}{Unif}
\nc{\Qu}{\ul{Q}}
\nc{\Qo}{\ol{Q}}
\nc{\Ro}{\ol{R}}
\nc{\Vu}{\ul{V}}
\nc{\Vo}{\ol{V}}
\nc{\RanQ}{\Delta Q}
\nc{\RanV}{\Delta V}
\nc{\clipQ}{\Delta \breve{Q}}
\nc{\frzQ}{\Delta \mathring{Q}}
\nc{\clipV}{\Delta \breve{V}}
\nc{\clipdelta}{\breve{\delta}}
\nc{\cliptheta}{\breve{\theta}}
\nc{\delmin}{\Delta_{{\rm min}}}
\nc{\delmins}[1]{\Delta_{{\rm min},{#1}}}
\nc{\gapfinal}[1]{\max \left\{ \frac{\frzQ_{{#1}}^{k^\st}(x,a)}{2H}, \frac{\delmin}{4H} \right\}}
\nc{\post}[2]{R({#1}; {#2})}
\nc{\posts}[3]{R_{#3}({#1}; {#2})}

\nc{\algnst}[1]{\begin{align*}#1\end{align*}}
\nc{\algn}[1]{\begin{align}#1\end{align}}
\nc{\matx}[1]{\left(\begin{matrix}#1\end{matrix}\right)}
\renewcommand{\^}[1]{^{(#1)}}

\nc{\nuu}{\nu}

\nc{\bel}[1]{\mathbf{b}({#1})}
\nc{\nbel}[1]{\bar{\mathbf{b}}({#1})}
\nc{\sbel}[2]{\mathbf{b}'_{#1}({#2})}
\nc{\nsbel}[2]{\bar{\mathbf{b}}'_{#1}({#2})}

\nc{\bv}{\mathbf{v}}
\nc{\bone}{\mathbf{1}}
\nc{\bX}{\mathbf{X}}
\nc{\bY}{\mathbf{Y}}
\nc{\bG}{\mathbf{G}}
\nc{\bz}{\mathbf{z}}
\nc{\bw}{\mathbf{w}}
\nc{\bA}{\mathbf{A}}
\nc{\bJ}{\mathbf{J}}
\nc{\bK}{\mathbf{K}}
\nc{\bb}{\mathbf{b}}
\nc{\ba}{\mathbf{a}}
\nc{\bc}{\mathbf{c}}
\nc{\bC}{\mathbf{C}}
\nc{\BR}{\mathbb R}
\nc{\BA}{\mathbb{A}}
\nc{\BC}{\mathbb C}
\nc{\bx}{\mathbf{x}}
\nc{\bS}{\mathbf{S}}
\nc{\bM}{\mathbf{M}}
\nc{\bR}{\mathbf{R}}
\nc{\bN}{\mathbf{N}}
\nc{\NN}{\mathbb{N}}
\nc{\by}{\mathbf{y}}
\nc{\sy}{y}
\nc{\sx}{x}

\nc{\MO}{\mathcal O}
\nc{\MU}{\mathcal{U}}
\nc{\ME}{\mathcal{E}}
\nc{\MN}{\mathcal{N}}
\nc{\MK}{\mathcal{K}}
\nc{\MM}{\mathcal{M}}
\nc{\MS}{\mathcal{S}}
\nc{\MT}{\mathcal{T}}
\nc{\BF}{\mathbb F}
\nc{\BQ}{\mathbb Q}
\nc{\MX}{\mathcal{X}}
\nc{\MA}{\mathcal{A}}
\nc{\MD}{\mathcal{D}}
\nc{\MB}{\mathcal{B}}
\nc{\MZ}{\mathcal{Z}}
\nc{\MJ}{\mathcal{J}}
\nc{\MW}{\mathcal{W}}
\nc{\MF}{\mathcal{F}}
\nc{\CF}{\mathcal{F}}
\nc{\MR}{\mathcal{R}}
\nc{\MY}{\mathcal{Y}}
\nc{\BZ}{\mathbb Z}
\nc{\BN}{\mathbb N}
\nc{\ep}{\epsilon}
\nc{\epbe}{\varepsilon_{\mathsf{BE}}}
\nc{\epout}{\varepsilon_{\mathsf{outlier}}}
\nc{\bellc}[1][h]{\MT_{#1}^\circ}
\nc{\vep}{\varepsilon}
\nc{\gapfn}[1]{\varepsilon_{#1}}
\nc{\ggapfn}[2]{\varphi_{#1}({#2})}
\nc{\epsahk}{\gapfn{0}}
\nc{\BH}{\mathbb H}
\nc{\BG}{\mathbb{G}}
\nc{\D}{\Delta}
\nc{\One}[1]{\mathbbm{1}\left\{{#1}\right\}}
\nc{\bOne}{\mathbf{1}}
\nc{\Aopt}{\mathcal{A}^{\rm opt}}
\nc{\Amul}{\mathcal{A}^{\rm mul}}

\nc{\SQ}{\mathsf Q}

\nc{\DO}{\accentset{\circ}{\D}}
\nc{\mf}{\mathfrak}
\nc{\mfp}{\mathfrak{p}}
\nc{\mfq}{\mf{q}}
\nc{\mfx}{\mf{s}}
\nc{\Sp}{\mbox{Spec}}
\nc{\Spm}{\mbox{Specm}}
\nc{\hookuparrow}{\mathrel{\rotatebox[origin=c]{90}{$\hookrightarrow$}}}
\nc{\hookdownarrow}{\mathrel{\rotatebox[origin=c]{-90}{$\hookrightarrow$}}}
\nc{\hra}{\hookrightarrow}
\nc{\tra}{\twoheadrightarrow}
\nc{\sgn}{{\rm sgn}}
\nc{\muideal}{\mu_{\mathsf{ideal}}}
\nc{\aut}{{\rm Aut}}
\nc{\Hom}{{\rm Hom}}
\nc{\img}{{\rm Im}}
\DMO{\id}{Id}
\DMO{\supp}{supp}
\DMO{\KL}{KL}
\nc{\kld}[2]{D_{\mathsf{KL}}({#1}||{#2})}
\nc{\ren}[2]{D_2({#1}||{#2})}
\nc{\chisq}[2]{\chi^2({#1}||{#2})}
\nc{\tvd}[2]{D_{\mathsf{TV}}({#1}, {#2})}
\nc{\hell}[2]{d_{\mathsf{H}}^2({#1}, {#2})}
\nc{\dbi}[3][\pi]{D_{\mathsf{bi}}^{#1}({#2} \| {#3})}
\DMO{\BSS}{BSS}
\DMO{\BES}{BES}
\DMO{\BGS}{BGS}
\DMO{\poly}{poly}
\nc{\indep}{\perp}
\DMO{\sink}{sink}
\nc{\fp}[1]{\MP_1({#1})}
\nc{\BO}{\mathbb{O}}
\nc{\BT}{\mathbb{T}}

\nc{\RR}{\mathbb{R}}
\nc{\Gradient}{\nabla}
\DMO{\diag}{diag}
\nc{\EE}{\mathbb{E}}
\nc{\MQ}{\mathcal{Q}}
\nc{\ML}{\mathcal{L}}
\nc{\cPhi}{\bar \Phi}

\DeclareMathOperator*{\PR}{Pr}
\renewcommand{\Pr}{\PR}
\nc{\E}{\mathbb{E}}
\nc{\ra}{\rightarrow}

\nc{\pmhc}[1]{\{-1,1\}^{#1}}
\nc{\Dbnd}{D}
\nc{\Bbnd}{B}

\nc{\Key}{\mathsf{KeyGen}}
\nc{\Enc}{\mathsf{Encode}}
\nc{\Encemb}{\mathsf{EncodeEmb}}
\nc{\Dec}{\mathsf{Decode}}
\nc{\sk}{\mathsf{sk}}
\nc{\pk}{\mathsf{pk}}
\nc{\lpk}{\ell_{\mathsf{pk}}}
\nc{\lsk}{\ell_{\mathsf{sk}}}
\nc{\msg}{\mathsf{m}}
\nc{\Adv}{\mathsf{Adv}}
\nc{\Red}{\mathsf{Red}}
\nc{\negl}{\mathsf{negl}}
\nc{\Ber}{\mathrm{Ber}}
\nc{\PRFPRC}{\mathsf{PRF\text{-}PRC}}
\nc{\wt}{\mathrm{wt}}
\nc{\res}[2]{{#1}_{#2}}
\nc{\bzero}{\mathbf{0}}
\nc{\Bin}{\mathrm{Bin}}
\nc{\Hyp}{\mathrm{Hyp}}

\nc{\Nrho}[1][\rho]{{N}_{#1}}
\nc{\Trho}[1][\rho]{\mathsf{T}_{#1}}
\nc{\hc}[1][n]{\{0,1\}^{#1}}
\nc{\Stab}{\mathbf{Stab}}
\nc{\bW}{\mathbf{W}}
\nc{\NS}{{\mathbf{NS}}}

\nc{\KeyS}{\mathsf{KeyGen_{Sub}}}
\nc{\EncS}{\mathsf{Encode_{Sub}}}
\nc{\DecS}{\mathsf{Decode_{Sub}}}
\nc{\WeightPerturb}{\mathsf{WeightPerturb}}
\nc{\Unique}{\mathsf{Unique}}
\nc{\PRCS}{\mathsf{PRC_{Sub}}}
\nc{\PRC}{\mathsf{PRC}}
\nc{\PRCI}{\mathsf{PRC_{Idx}}}
\nc{\SampleUnique}{\mathsf{SampleUnique}}
\nc{\PerturbDifference}{\mathsf{PerturbDifference}}

\nc{\Model}{\mathsf{Model}}
\nc{\Modelo}{\overline{\Model}}
\nc{\prompt}{\mathtt{PROMPT}}
\nc{\Setup}{\mathsf{Setup}}
\nc{\Detect}{\mathsf{Detect}}
\nc{\Sigprc}{\Sigma_{\mathsf{PRC}}}
\nc{\Wat}{\mathsf{Wat}}
\nc{\term}{\mathtt{END}}
\nc{\tok}{\mathsf{t}}
\nc{\True}{\textsf{True}}
\nc{\False}{\textsf{False}}
\nc{\Eemb}{\ME_{\mathsf{Emb}}}
\nc{\hist}{\mathsf{hist}}
\nc{\hh}{\mathsf{h}}
\nc{\freq}{\mathsf{freq}}
\nc{\ff}{\mathsf{f}}

\nc{\Hemp}[1]{H_{\mathsf{e}}^{#1}}
\nc{\Hempt}[1]{\bar{H}_{\mathsf{e}}^{#1}}
\nc{\Hemptil}[1]{\tilde{H}_{\mathsf{e}}^{#1}}
\nc{\Hmean}[1]{H_{\mathsf{m}}^{#1}}
\nc{\partition}[1][n,q]{P^{\mathsf{ptn}}_{#1}}
\nc{\Crob}{C_{\mathsf{rob}}}
\nc{\Lmax}{L_{\mathsf{max}}}
\nc{\skwat}{\sk_{\mathsf{Wat}}}
\nc{\EmbedToken}{\mathsf{EmbedChar}}
\nc{\len}{\mathrm{len}}
\nc{\Esub}{\ME_{\mathsf{sub}}}
\nc{\Ecomp}{\ME_{\mathsf{comp}}}
\nc{\comp}{\mathsf{c}}
\nc{\SE}{\mathscr{E}}
\nc{\alphb}{q}
\nc{\tAdv}{\widetilde{\Adv}}
\nc{\Funif}{{F_{\mathsf{Unif}}}}
\nc{\Alg}{\mathsf{Alg}}
\nc{\Majority}{\mathsf{Maj}}
\nc{\Dist}{\mathsf{Dist}}
\nc{\edit}{edit\xspace}
\nc{\Edit}{Edit\xspace}
\nc{\Wcomp}{\MW^{\mathsf{comp}}}

\nc{\INS}{\mathsf{INS}}
\nc{\CNS}{\mathsf{CNS}}
\nc{\cdist}{\stackrel{\mathrm{c}}{\sim}}
\nc{\SU}{\mathscr{U}}
\nc{\rr}{\bar{n}}

\nc{\KeyGen}{\mathsf{KeyGen}}
\nc{\ED}{D_{\mathsf{ED}}}
\nc{\Ham}{D_{\mathsf{Ham}}}
\nc{\bin}{\mathsf{bin}}
\nc{\EDball}{\mathcal{B}_{\mathsf{ED}}}
\nc{\SEDball}{\mathcal{B}_{\mathsf{Ham,ED}}}
\nc{\LEDball}{\mathcal{B}_{\mathsf{len,ED}}}
\nc{\epED}{\varepsilon_{\mathsf{ED}}}
\nc{\epDec}{\varepsilon_{\mathsf{Dec}}}
\nc{\Eedit}{\mathscr{E}^{\mathsf{edit}}}
\nc{\Egood}{\ME^{\mathsf{good}}}
\nc{\Ebad}{\ME^{\mathsf{bad}}}
\nc{\PermEnc}{\mathsf{PermEncode}}
\nc{\dham}{d_{\mathsf{H}}}
\nc{\dedit}{d_{\mathsf{E}}}
\nc{\pDec}{p_{\mathsf{Dec}}}
\nc{\Rclean}{R_{\mathsf{Clean}}}
\nc{\adv}{\mathcal{A}}
\nc{\bi}{\mathbf{i}}

\nc{\bj}{\mathbf{j}}
\nc{\bp}{\mathbf{p}}
\nc{\Ologit}{\MO_{\mathsf{logit}}}
\nc{\Osamp}{\MO_{\mathsf{samp}}}
\nc{\epapx}{\varepsilon_{\mathsf{apx}}}
\nc{\DistSpanner}{\textsf{DistSpanner}\xspace}
\nc{\epbase}{\vep_{\mathsf{base}}}
\nc{\cnorm}{\beta}
\nc{\occ}{\mathsf{occ}}

\nc{\phiprev}{\phi^{\mathsf{prev}}}
\nc{\depth}{\mathsf{depth}}
\nc{\range}{\mathsf{range}}
\nc{\Spread}{\mathsf{Spread}}

\nc{\Ebase}{\ME^{\mathsf{base}}}
\nc{\Eincr}{\ME^{\mathsf{incr}}}
\nc{\MIN}{\mathsf{MIN}}
\nc{\MAX}{\mathsf{MAX}}

\nc{\thrup}{\tau_{\mathsf{upp}}}
\nc{\thrlo}{\tau_{\mathsf{low}}}

\nc{\Sact}{\MS^{\mathsf{act}}}
\nc{\Sterm}{\MS^{\mathsf{term}}}
\nc{\mualg}{\mu_{\mathsf{alg}}}
\nc{\Erej}{\ME^{\mathsf{rs}}}
\nc{\Eterm}{\ME^{\mathsf{term}}}
\nc{\N}{\mathsf{N}}
\nc{\Lconst}{\mathsf{L}}
\nc{\tilmualg}{\tilde{\mu}^{\mathsf{alg}}}
\nc{\Gret}{\MG^{\mathsf{ret}}}
\nc{\Nint}{\MN^{\mathsf{int}}}
\nc{\Nappend}{\MN^{\mathsf{append}}}
\nc{\Nfinal}{\MN^{\mathsf{fin}}}
\nc{\prev}{\mathfrak{p}}
\nc{\Erestrict}{\ME^{\mathsf{restrict}}}
\nc{\Vinit}{V^{\mathsf{init}}}
\nc{\vinit}{v^{\mathsf{init}}}
\nc{\query}{\mathsf{query}}
\nc{\key}{\mathsf{key}}
\nc{\base}{\mathsf{base}}
\nc{\Nvalid}{N_{\mathsf{valid}}}
\nc{\Vvalid}{\MV^{\mathsf{valid}}}
\nc{\Equery}{\ME^{\mathsf{query}}}
\nc{\tAlg}{\widetilde{\Alg}}
\nc{\piref}{\pi_{\mathsf{ref}}}
\nc{\Sym}{\mathrm{Sym}}

\nc{\BiLinVI}{\textsf{BiLinVI}}
\nc{\AL}{A}
\nc{\AR}{B}
\nc{\aL}{a}
\nc{\aR}{b}

\nc{\actL}{\mathfrak{a}_{\mathsf{L}}}
\nc{\actR}{\mathfrak{a}_{\mathsf{R}}}

\nc{\VL}{V_{\mathsf{L}}}
\nc{\VR}{V_{\mathsf{R}}}
\nc{\MAL}{\mathcal{A}_{\mathsf{L}}}
\nc{\MAR}{\mathcal{A}_{\mathsf{R}}}

\nc{\piL}{\pi_{\mathsf{L}}}
\nc{\piR}{\pi_{\mathsf{R}}}
\nc{\uL}{u_{\mathsf{L}}}
\nc{\uR}{u_{\mathsf{R}}}
\nc{\pL}{p_{\mathsf{L}}}
\nc{\pR}{p_{\mathsf{R}}}
\nc{\hatpL}{\widehat{p}_{\mathsf{L}}}
\nc{\hatpR}{\widehat{p}_{\mathsf{R}}}
\nc{\barpR}{\bar{p}_{\mathsf{R}}}
\nc{\barpL}{\bar{p}_{\mathsf{L}}}
\nc{\cL}{c_{\mathsf{L}}}
\nc{\cR}{c_{\mathsf{R}}}
\nc{\tilpL}{\tilde{p}_{\mathsf{L}}}
\nc{\tilpR}{\tilde{p}_{\mathsf{R}}}

\nc{\PCP}{\textsf{PCP}\xspace}
\nc{\PPAD}{\textsf{PPAD}\xspace}
\nc{\EOTL}{\textsc{End of The Line}}
\nc{\Lp}{\textsf{L}\xspace}
\nc{\Rp}{\textsf{R}\xspace}

\title{The Fine-Grained Complexity of Approximate Nash Equilibrium and Free Games}
\author{Noah Golowich\thanks{Email: \url{nzg@cs.utexas.edu}}}
\date{\today}

\begin{document}
\maketitle

\begin{abstract}
We study the fine-grained complexity of computing approximate Nash equilibria and approximating the value of free games
in the regime where the approximation error vanishes. Under the \PCP for \PPAD and ETH
for \PPAD conjectures, we show that computing $\ep$-approximate Nash equilibria in 2-player $N$-action normal-form games requires time $N^{(\log(N)/\ep^2)^{1-o(1)}}$, thus showing that the classical Lipton-Markakis-Mehta algorithm \cite{LMM03} is optimal through all regimes of $\ep = \omega(1/\sqrt{N})$. While such optimality was known in the constant-$\ep$ regime \cite{Rub16}, previous work could only rule out significantly smaller running times of $N^{O(\log(N)/\ep)}$ in the regime $\ep = o(1)$.  Using similar techniques, we then establish an analogous tight
lower bound of $N^{(\log(N)/\ep^2)^{1-o(1)}}$ under ETH for $\ep$-additive value estimation in free games,
when $\ep\geq 2^{-o(\sqrt{\log n})}$, answering a question of \cite{AIM14}.
\end{abstract}

\paragraph{Statement on AI Use.}
The results written in this paper were generated with the use of GPT-5.6-sol ultra. %
Please see \cref{sec:ai-use} for a more detailed description of the role of AI in the development of this paper.

\section{Introduction}

Finding a Nash equilibrium of a two-player normal-form game is a
fundamental total search problem. In this paper we focus on the approximate version of the problem, parametrized by some $\ep \in (0,1)$ representing the maximum additive violation to the players' constraints. When $\ep$ is an inverse polynomial in the game size, computing $\ep$-approximate Nash equilibria is known to be 
\PPAD-complete~\cite{CDT09,DGP09}, while the problem admits a
quasipolynomial-time algorithm when $\ep$ is a constant. In particular, for games with payoffs normalized to $[0,1]$,
Lipton, Markakis, and Mehta showed that there is an $\ep$-approximate Nash equilibrium
in which each player has support of size $O((\log N)/\ep^2)$, yielding an
$N^{O((\log N)/\ep^2)}$-time algorithm by exhaustive
search~\cite{LMM03}. At constant accuracy this running time is
quasipolynomial; as $\ep$ decreases, however, the quadratic dependence of the exponent on
$1/\ep$ becomes the dominant cost. 

When $\ep$ is a sufficiently small constant, it is known that we cannot hope for algorithms for $\ep$-approximate Nash equilibria running in time $N^{o(\log N)}$, under the \PCP for \PPAD and ETH for \PPAD conjectures \cite{BPR15}. %
A  straightforward extension of the techniques of \cite{BPR15} suffices to show a lower bound of $N^{o(\log(N)/\ep)}$ for $1/N \ll \ep \leq o(1)$, which is off by a factor of $1/\ep$ in the exponent when compared with the upper bound of \cite{LMM03}.

In light of this discussion, a natural question is: can one obtain a fixed polynomial improvement in the exponent
$(\log N)/\ep^2$ from the result of \cite{LMM03}? Our first main result gives a negative answer throughout the
vanishing-$\ep$ regime, under standard computational
assumptions.

\begin{theorem}[Nash lower bound]
\label{thm:main-ne}
Fix any constant $c > 0$ and any slowly decreasing function (\cref{def:slowly-decreasing}) $N \mapsto \ep(N)$ satisfying $\ep(N) = \omega(1/\sqrt{N}), \ep(N) = o(1)$. Under the \PCP for \PPAD and ETH for \PPAD conjectures (\cref{conj:pcp-ppad}), there is no algorithm which finds a $\ep(N)$-Nash equilibrium in 2-player, $N$-action normal-form games with payoffs in $[0,1]$ in time $N^{((\log N)/\ep(N)^{2})^{1-c}}$. 
\end{theorem}

\paragraph{Free games.} A related approximation problem arises for \emph{free games},
two-prover interactive proofs in which the verifier sends the provers independent
questions.  Given an explicit free game (namely, one where the input takes the form of a table encoding the verifier function), the computational problem $\ep$-\textsc{FreeGame} asks for an
additive-$\ep$ estimate of its value. Writing $N$ for the number of entries
in the verifier table, Aaronson, Impagliazzo, and Moshkovitz \cite{AIM14} gave an
$N^{O((\log N)/\ep^2)}$-time algorithm, which proceeds by subsampling a logarithmic-sized set of questions. This sampling procedure is analogous to the result of \cite{LMM03} which proceeds by subsampling logarithmic-sized sets of actions of a normal-form game. 

On the lower bound side, \cite{AIM14} gave an 
$N^{\widetilde\Omega((\log N)/\ep)}$ lower bound under ETH, and asked
whether the gap in the dependence on $\ep$ could be closed~\cite{AIM14}. This gap was recently highlighted by Bernasconi et al., \cite{BCCF26}, who introduced a unified framework for finding small covers of families of low-degree polynomials defined over polytopes, which allows them to recover a number of $N^{O(\log(N)/\ep^2)}$-time algorithms based roughly on ``sparsification'' ideas, including those for $\ep$-approximate Nash \cite{LMM03} and $\ep$-\textsc{FreeGame} \cite{AIM14}. 

Our second main result shows that the inverse-quadratic dependence on $\ep$ for $\ep$-\textsc{FreeGame} is also optimal:
\begin{theorem}[Free game lower bound]
\label{thm:freegame-main}
There is a sufficiently small absolute constant $\alpha^\star>0$ such that
the following holds.  Fix any constant $c>0$ and any slowly decreasing function (\cref{def:slowly-decreasing})
function
$N\mapsto\ep(N)>0$ such that $\ep(N)=o(1)$ and
$\ep(N)\geq2^{-\alpha^\star\sqrt{\log N}}$ for all sufficiently large $N$.
Under \cref{conj:eth-freegame}, there is no uniform deterministic algorithm
which, given any explicit free game of size $N$, computes an
additive-$\ep(N)$ estimate of its value in time
\begin{align}
N^{((\log N)/\ep(N)^2)^{1-c}}.
\label{eq:freegame-forbidden-runtime}
\end{align}
\end{theorem}
Since free games may be interpreted as dense constraint satisfaction problems (CSPs) with polynomial-sized alphabets, the lower bound of \cref{thm:freegame-main} yields a corresponding lower bound for additive approximation of such CSPs. 

\paragraph{Sketch of technical ideas.} The main challenge in the proof of \cref{thm:main-ne} (which parallels those in proving \cref{thm:freegame-main}) is the loss incurred by the usual birthday-repetition
reduction~\cite{BPR15}. Typically, in this reduction, both players encode blocks of $k$ vertices from a hard instance of computing Nash equilibria in a bipartite polymatrix game with some number $m$ of players.
Then a pair of blocks captures only
$O(k^2/m)$ edges, which ultimately leads to a lower bound of $N^{\log(N)/\ep}$, asymptotically smaller than $N^{\log(N)/\ep^2}$ when the desired accuracy $\ep(N)$ is $o(1)$ (corresponding to the regime $k = o(\sqrt{m})$). We instead use an asymmetric partitioning into blocks (i.e., where players' blocks are different sizes) combined with a construction generalizing the Hadamard code to improve the exponent in the lower bound from $\log(N)/\ep$ to $\log(N)/\ep^2$.
The free-game result in \cref{thm:freegame-main} reuses this idea of asymmetric block sizes, but
needs additional ideas resulting from the fact that in free games it is not possible for one player to ``punish'' the other as in normal-form games.  %

\section{Preliminaries}
\label{sec:prelims}
In this section, we introduce preliminaries needed for our main results. First, we introduce the following notion of \emph{slowly decreasing}, which facilitates some technical details in the proofs of \cref{thm:freegame-main,thm:main-ne}.
\begin{definition}[Slowly decreasing]
\label{def:slowly-decreasing}
A function $f : \BN \to (0,1)$ is \emph{slowly decreasing} if it is  efficiently computable, non-increasing and there is a constant $c > 0$ so that for all $N$, $f(N+1) \geq c \cdot f(N)$. 
\end{definition}

\subsection{Normal-form games}
Our focus when considering normal-form games is on the setting where there are 2 players. We refer to the two players as $\mathsf{L}$ and $\mathsf{R}$ (``left'' and ``right''). Each player has \emph{action set} $[N]$; the game is then specified by utility functions $\uL, \uR : [N] \times [N] \to \BR$.

For an error parameter $\ep > 0$, an \emph{$\ep$-Nash equilibrium} of the game $(\uL, \uR)$ is a tuple $(\pL, \pR) \in \Delta([N])^2$ satisfying
\begin{align}
    \uL(\pL, \pR) \geq \max_{a' \in [N]} \uL(a', \pR)-\ep, \qquad \uR(\pL, \pR) \geq \max_{a' \in [N]} \uR(\pL, a')-\ep\nonumber,
\end{align}
where we write $\uL(\pL,\pR) = \E_{a \sim \pL, a' \sim \pR}[\uL(a,a')]$, and similarly for $\uR(\pL, \pR), \uL(a', \pR), \uR(\pL, a')$. 

\subsection{\PCP for \PPAD and ETH for \PPAD conjectures}
To state the conjectures upon which \cref{thm:main-ne} rests, we first introduce \emph{graphical games}. An $m$-player graphical game is specified by a graph $G$ with vertex set $[m]$ and edge set $E$. Each player $i \in [m]$ has an action set $\MA_i$, and for each ordered pair $(i,j)$ with $\{i,j\} \in E$, there is an associated utility function $u_{ij} : \MA_i \times \MA_j \to [0,1]$. For an action profile $\ba = (\ba_1, \ldots, \ba_m) \in \prod_{i=1}^m \MA_i$, the utility of player $i$ is then given by
\begin{align}
u_i(\ba) = \sum_{j:\ \{i,j\} \in E} u_{ij}(\ba_i, \ba_j)\nonumber.
\end{align}
For values $\ep, \delta > 0$, an \emph{$(\ep, \delta)$-weak Nash equilibrium} in a graphical game is a tuple $p = (p_1, \ldots, p_m) \in \prod_{i \in [m]} \Delta(\MA_i)$ so that for at least $(1-\delta) \cdot m$ players $i$, we have
\begin{align}
\E_{\ba \sim p}[u_i(\ba)] \geq \max_{a_i' \in \MA_i} \E_{\ba \sim p}[u_i(a_i', \ba_{-i})] - \ep\nonumber.
\end{align}

Our results rely on the following conjecture which is a consequence of the combination of the \PCP for \PPAD conjecture and the ETH for \PPAD conjecture~\cite{BPR15,Rub16}.
\begin{conjecture}[\PCP for \PPAD and ETH for \PPAD]
\label{conj:pcp-ppad}
There are constants $\ep, \delta > 0$ so that finding an $(\ep,\delta)$-weak Nash equilibrium in a bipartite polymatrix game with $m$ players on each side, maximum degree 3, and 2 actions per player takes time $2^{\tilde \Omega(m)}$.
\end{conjecture}

\subsection{\BiLinVI{} problem}
In place of directly reducing from \cref{conj:pcp-ppad}, it will be more convenient to reduce from the following two-sided linear variational-inequality problem, which we call \BiLinVI{}. It is very similar to the \textsf{LinVI} problem studied in \cite{APSY26}.
\begin{definition}[\BiLinVI]
A \BiLinVI{} instance is specified by an integer $m \in \BN$ and matrices $\AL, \AR \in [-1,1]^{m \times m}$ and vectors $\aL, \aR \in [-1,1]^m$. For an error parameter $\eta$, the problem is to find vectors $\bx, \by \in [-1,1]^m$ for which
\begin{align}
\max_{\bx' \in [-1,1]^m} \frac{1}{2m} \left\langle \bx' - \bx, \AL \by + \aL \right\rangle  &\leq \eta,\label{eq:bilinx}\\
\max_{\by' \in [-1,1]^m} \frac{1}{2m} \left\langle \by' - \by, \AR \bx + \aR \right\rangle &\leq \eta.\label{eq:biliny}
\end{align}
\end{definition}

\begin{lemma}
    \label{lem:bilinvi-hard}
Under \cref{conj:pcp-ppad}, there is an absolute constant $\eta$ so that finding an $\eta$-approximate solution to a \BiLinVI{} instance takes $2^{\tilde \Omega(m)}$ time, even if the matrices $\AL, \AR$ are promised to have at most 3 nonzero entries in every row and column.
\end{lemma}

\subsection{Free games, bipartite 2-CSPs, and the Exponential Time Hypothesis}

A \emph{free game} is a tuple
$\mathcal F=(X,Y,A,B,V)$ consisting of finite question sets $X,Y$, finite
answer sets $A,B$, and a verifier
$V:X\times Y\times A\times B\to[0,1]$.  The verifier draws $(x,y)$ uniformly
from the product set $X\times Y$, and the value of the game is
\begin{align}
\omega(\mathcal F)
:=\max_{f_{\mathsf L}:X\to A,\,f_{\mathsf R}:Y\to B}
\frac{1}{|X||Y|}\sum_{x\in X,\,y\in Y}
V(x,y,f_{\mathsf L}(x),f_{\mathsf R}(y)).
\label{eq:freegame-value}
\end{align}
The problem $\ep$-\textsc{FreeGame} asks for a number within additive
$\ep$ of $\omega(\mathcal F)$.  We measure the size of an explicit free
game by
$
n:=|X||Y||A||B|,
$
which denotes the number of entries in the verifier table; the games in our
reduction have constant-bit entries.

A \emph{bipartite $2$-CSP} is a constraint-satisfaction problem in which
each constraint involves one variable from each of two sides.  Formally, an
instance is a tuple
\begin{align}
\mathcal I
:=\big(\VL,\VR,E,\Sigma_{\mathsf L},\Sigma_{\mathsf R},
(P_e:\Sigma_{\mathsf L}\times\Sigma_{\mathsf R}\to\{0,1\})_{e\in E}\big),
\label{eq:bipartite-2csp}
\end{align}
where $\VL$ and $\VR$ are finite sets of variables,
$\Sigma_{\mathsf L}$ and $\Sigma_{\mathsf R}$ are their respective finite
label alphabets, and $E$ is a nonempty multiset of constraint occurrences
$e=(i,j)\in\VL\times\VR$.  A labeling consists of maps
$a:\VL\to\Sigma_{\mathsf L}$ and
$b:\VR\to\Sigma_{\mathsf R}$; it satisfies an occurrence $e=(i,j)$ exactly
when $P_e(a(i),b(j))=1$.  The value of the instance is the largest fraction
of simultaneously satisfied constraint occurrences:
\begin{align}
\operatorname{val}(\mathcal I)
:=\max_{a:\VL\to\Sigma_{\mathsf L},\,b:\VR\to\Sigma_{\mathsf R}}
\frac{1}{|E|}\sum_{e=(i,j)\in E}P_e(a(i),b(j)).
\label{eq:bipartite-2csp-value}
\end{align}
In particular, $\operatorname{val}(\mathcal I)=1$ means that one labeling
satisfies every constraint.  We allow parallel occurrences with the same
endpoints, and count them separately in $|E|$ and when measuring degrees.  The \emph{support graph} of $\mathcal I$ is the bipartite
graph containing $(i,j)$ whenever at least one constraint occurrence has
those endpoints.  We say that $\mathcal I$ is $D$-regular if every variable
is incident to exactly $D$ constraint occurrences, counted with
multiplicity.

A Boolean-valued free game can therefore be viewed as a dense bipartite
$2$-CSP: the questions are the variables, the answers are their labels, and
there is a constraint for every pair of questions.  A verifier with values in
$[0,1]$ is the corresponding real-valued version of this constraint system.

\paragraph{ETH and consequences.} We use the following standard sparse formulation of ETH (see \cite{IPZ01}).  %

\begin{conjecture}[Exponential Time Hypothesis]
\label{conj:eth-freegame}
There is an absolute constant $C_{\mathsf{sat}}\geq1$ such that no uniform
deterministic algorithm decides satisfiability of $3$-CNF formulas with $r$
variables and at most $C_{\mathsf{sat}}r$ clauses in time $2^{o(r)}$.
\end{conjecture}

Our reduction establishing hardness for free games will use the the following fixed-soundness consequence of the PCP theorem.  Its proof is deferred to
\cref{sec:free-game}.

\begin{lemma}[Fixed-soundness normalization]
\label{lem:freegame-source-normalization}
Let $\sigma_\star:=2^{-24}$ and $\eta > 0$ be an arbitrary constant.  Under
\cref{conj:eth-freegame}, there is no uniform deterministic algorithm which takes as input a bipartite $D$-regular $2$-CSP instance $\mathcal{I}$ of the following form: %
\begin{align}
\mathcal I
=\big(\VL,\VR,E,\Sigma_{\mathsf L},\Sigma_{\mathsf R},
 (P_e:\Sigma_{\mathsf L}\times\Sigma_{\mathsf R}\to\{0,1\})_{e\in E}\big),
\qquad
|\VL|=|\VR|=m,\qquad |E|=Dm,
\label{eq:freegame-pcp-instance}
\end{align}
where $D\in\mathbb{N}$, $q_{\mathsf L}:=|\Sigma_{\mathsf L}|$, and
$q_{\mathsf R}:=|\Sigma_{\mathsf R}|$ are absolute constants, runs in time $2^{O(m^{1-\eta})}$, and distinguishes between the cases $\operatorname{val}(\mathcal I)  =1$ and $\operatorname{val}(\mathcal I)\leq\sigma_\star$.  
\end{lemma}

\section{Technical Overview}
We prove \cref{thm:main-ne} by reducing the hard \BiLinVI{} problem from
\cref{lem:bilinvi-hard} to computing approximate Nash equilibrium in
2-player normal-form games. The reduction is parameterized by the 
dimension $m$ from a \BiLinVI{} instance (which we refer to as the \emph{source instance}) and an integer $1\leq t\leq\sqrt m$. It produces a game with
at most $2^{O(t\log m)}$ actions per player in which every $\ep$-Nash equilibrium yields an
$O(\ep\sqrt m/t)$-approximate \BiLinVI{} solution. Thus choosing
$t=\Theta(\ep\sqrt m)$ keeps this error below the constant hardness
threshold; the claimed result follows by choosing $m$ so that the action sets fit
within the target size $N$.
We first explain why the usual
birthday-repetition approach from prior work~\cite{BPR15} does not give the desired dependence on $\ep$.
We then describe how the present construction uses asymmetric blocks together with a Hadamard code-like construction to avoid this limitation. 

\paragraph{Why birthday repetition is insufficient.}
We first describe the usual ``birthday repetition'' argument which can be used to show that $\ep$-Nash equilibrium requires time $N^{\tilde \Omega(\log(N)/\epsilon)}$. 
Fix a \BiLinVI{} instance $(\AL, \AR, \aL, \aR)$; it will be useful to consider the bipartite graph $G$ whose vertices on each side are in correspondence with $[m]$, and where there is an edge $(i,j)$ if $\AL_{ij} \neq 0$ or $\AR_{ji} \neq 0$. We use that every row and column of $\AL,\AR$ has at most 3 nonzero entries (as in \cref{lem:bilinvi-hard}), which implies that $G$ has maximum degree $6$.
In the standard birthday-repetition approach~\cite{BPR15}, one converts this \BiLinVI{} instance into a normal form game $\MG$ with $m^{\Theta(k)}$ actions, for some parameter $k$.  The game $\MG$ has two players, $\mathsf{L}, \mathsf{R}$, corresponding to the two sides of the bipartite graph $G$. An action of the $\mathsf{L}$-player in $\MG$ specifies a block $I \subset [m]$ of
$k$ coordinates in $[m]$ together with a value in $\{-1,1\}^I$, corresponding to a choice of bit for each coordinate of $I$. Thus, a distribution over a set of actions all focused on a single block $I$ corresponds to a real-valued assignment $\bx_I \in [-1,1]^I$ of coordinates in $I$. If one partitions $[m]$ into such blocks $I$ and mixes such a distribution for each block $I$, it is possible to read off a full assignment $\bx \in [-1,1]^m$; similar considerations apply for blocks $J \subset [m]$ and assignments $\by_J$ for player-$\mathsf{R}$. 

In order for the expected payoffs to \textsf{L}, \textsf{R} in $\MG$ to line up with the corresponding values of \cref{eq:bilinx,eq:biliny}, respectively, we define the payoffs of $\MG$ as follows: if \textsf{L}-player plays $(I, \bx_I)$ and \textsf{R}-player plays $(J, \by_J)$, then the payoff to $\Lp$ is the portion of \cref{eq:bilinx} corresponding to the subsets of indices $I,J$, namely $\frac{1}{|I|} \sum_{i \in I} \aL_i \cdot \bx_i + \sum_{i \in I, j \in J} \AL_{ij} \cdot \bx_i \by_j$. Payoffs to \Rp are defined analogously. %

If we imagine that such blocks $I \subset [m], J \subset [m]$ are drawn uniformly and independently at random, then such a tuple $(I,J)$ contains $O(k^2/m)$ edges in expectation. Thus payoffs in $\MG$ are ``scaled down'' from those in the left-hand sides of \cref{eq:bilinx,eq:biliny} by a factor of  $O(k^2/m)$ (assuming the players are mixing approximately uniformly over $I,J$, which will be the case in approximate Nash equilibria). Thus, roughly speaking, in order for an $\ep$-Nash
equilibrium of $\MG$ to imply that the inequalities \cref{eq:bilinx,eq:biliny} hold for the induced vectors $\bx,\by$, for some absolute constant $\eta$, one must have $\ep \leq O(\eta k^2/m) = O(k^2/m)$, i.e., 
$
k=\Omega(\sqrt{m\ep}).
$
On the other hand, this value of $k$ yields $\log N=\Theta(k \log m) = \widetilde\Theta(k)$, and hence
$
m=\widetilde\Theta\left(\frac{(\log N)^2}{\ep}\right).
$
Consequently, the $2^{\widetilde\Omega(m)}$ lower bound on the complexity of \BiLinVI{} (from \cref{lem:bilinvi-hard}) translates to a complexity of 
$
2^{\widetilde\Omega((\log N)^2/\ep)}
=N^{\widetilde\Omega((\log N)/\ep)}
$ 
for computing $\ep$-Nash equilibrium of $\MG$, which has exponent that is asymptotically smaller than the best-known upper bound of $N^{O(\log(N)/\ep^2)}$. 

\paragraph{Our approach: asymmetric blocks.} At a high level, the gap in the approach described above results from the fact that the graph $G$ is very sparse, which forces the blocks to be very large in order to make the ``scaling factor'' $O(k^2/m)$ large enough, which in turn blows up the number of actions $N = m^{\Theta(k)}$. \emph{Is there a way to achieve large ``scaling factor'' without ``spending'' so many actions of the \textsf{L}, \textsf{R} players of $\MG$?}

To achieve this goal, it turns out to be useful to use \emph{asymmetric} block sizes: in particular, for integers $t,h$ satisfying $t \ll h$ and $m \leq t\cdot h < 2m$, we will partition $[m]$ into $t$ blocks $I \subset [m]$ of size roughly $h$ for the \textsf{L}-player and partition $[m]$ into $h$ blocks $J \subset [m]$ of size roughly $t$ for the \textsf{R}-player. It is straightforward to show (see \cref{lem:partitions}) that it is possible to choose such partitions so that each pair of blocks includes a constant number of edges of $G$ (which conforms with what one obtains by drawing blocks of sizes $h,\ t$, respectively at random); this fact keeps the payoffs of $\MG$ bounded by a constant.

The key insight is as follows: the standard birthday repetition approach described above would prescribe an action of \Lp which corresponds to each vector in $\{-1,1\}^I$, which would require $2^h$ actions for \Lp-player. Instead, actions of \Lp will correspond to tuples $(I, v)$, where $I \subset [m]$ is one of the $t$ blocks of size $h$, and $v \in \MV(h,t) \subseteq \{-1,1\}^h$ is a smaller ``chosen set'' of vectors (specified in \cref{lem:scaled-hypercube}) whose convex hull contains all elements of $[-1,1]^h$ without any ``too-large'' entries. In particular, we have $|\MV(h,t)| = 2^{\tilde \Theta(t)} \ll 2^h$ and
\begin{align}
\rho[-1,1]^h\subseteq\operatorname{conv}(\MV(h,t)),
\qquad \mbox{ where } \ \ 
\rho:=\frac14\sqrt{\frac{t}{h}}=\Theta\left(\frac{t}{\sqrt m}\right).\label{eq:v-prop-intro}
\end{align}
Actions of \Rp will correspond to tuples $(J,z)$ where $J \subset [m]$ is one of the $h$ blocks of size $t$ (for \Rp-player), and $z \in \{-1,1\}^t$, exactly as in the standard birthday repetition argument (which is not problematic since blocks of \Rp are smaller than blocks of \Lp). The payoff to \Lp if they choose $(I,v)$ and \Rp chooses $(J,z)$ is analogous to above, namely 
\begin{align}
\frac{1}{h} \sum_{i \in I} \aL_i \cdot v_i + \sum_{i \in I, j \in J} \AL_{ij} \cdot v_i \cdot z_j,\label{eq:mg-payoffs}
\end{align}
where we think of $v$ as indexed by $I$ and $z$ as indexed by $J$; payoffs for \Rp are defined symmetrically. The main idea here is that taking $\MV(h,t)$ to be smaller than $\{-1,1\}^h$ will not hurt us if we can ensure that the \Lp player never has incentive to choose a mixed strategy whose corresponding vector deviates outside of the ``small-entry'' cube $\rho[-1,1]^h$ (as we will do below).  

Now suppose that \Lp chooses a distribution over actions $(I,v)$ where the $I$-component is fixed and $v$ has expectation given by $\rho \cdot \bx_I \in \rho \cdot [-1,1]^h$ for some $\bx_I \in [-1,1]^h$, and that \Rp chooses a distribution over  actions $(J,z)$ whose $J$-component is fixed and $z$ has expectation given by $\by_J \in [-1,1]^t$. Then their payoff is ``scaled down'' from that in \cref{eq:bilinx} by a factor of $\Theta((t\cdot h/m) \cdot \rho) = \Theta(\rho)= \Theta(t/\sqrt m)$. By way of comparison, if we had used the birthday repetition approach with blocks of size $t$, we would have had a scaling factor of $\Theta(t^2/m) \ll t/\sqrt m$.
In summary, compressing the larger left
blocks therefore preserves a larger scale while keeping the action
descriptions short.

\paragraph{Enforcing a valid representation.}
In order for the reduction described above to work, we need that the $\ep$-Nash equilibrium of $\MG$ have the following two properties: \begin{enumerate}
\item For each player \Lp, \Rp, the distribution over their ``block index'' (i.e., $I$ for \Lp-player and $J$ for \Rp-player) must be roughly uniform over their set of blocks.
\item For player \Lp, for each block $I$, conditioned on choosing an action of the form $(I,v)$ for some $v \in \MV(h,t)$, the expectation of $v$ must lie in $\rho \cdot [-1,1]^h$. 
\end{enumerate}
The first property above is standard in previous birthday-repetition arguments~\cite{BPR15}, and it implies that the payoff to \Lp in $\MG$ (as determined by averaging \cref{eq:mg-payoffs} over the blocks $I,J$) is close to that in \cref{eq:bilinx} (and similarly for \Rp and \cref{eq:biliny}).

The second property is unique to our argument and is needed to ensure that we obtain a \emph{valid} vector $\bx \in [-1,1]^m$ from an $\ep$-Nash equilibrium of $\MG$. In particular, the conversion from an $\ep$-Nash in $\MG$ to a purported solution $(\bx, \by)$ of \cref{eq:bilinx,eq:biliny} constructs a vector $\rho \cdot \bx_I \in \rho \cdot [-1,1]^h$ for each block $I$, as described above. But because $\mathrm{conv}(\MV(h,t))$ is a \emph{strict superset} of $\rho \cdot [-1,1]^h$, it is possible that the resulting $\bx_I$ does not belong to the solid hypercube $[-1,1]^I$. 

To ensure that the above properties hold in any $\ep$-Nash equilibrium of $\MG$, we equip both players with additional actions, referred to as ``secondary actions'' (in contrast, we refer to actions of the form $(I,v)$ and $(J,z)$ discussed above as ``primary actions''). Player-\Rp has secondary actions which punish player-\Lp if either (a) their action distribution is not near-uniform over the block index $I$ or (b) if there are too many blocks $I$ for which the conditional expectation of $v$ is not in $\rho \cdot [-1,1]^h$ (per the second item above). Similarly, Player-\Lp has secondary actions which punish \Rp if their action distribution is not near-uniform over the block index $J$.

We remark that the secondary actions used to ensure the first item above are similar to the Althofer-style gadget used in prior birthday-repetition reductions~\cite{BPR15}. One difference is that if we were to follow that construction exactly, the number of Player-\Lp secondary actions would be $2^{\tilde O(h)}$ as Player-\Rp has $h$ blocks. To avoid this, player-\Lp secondary actions use a second copy $\MW$ of $\MV(h,t)$, whose convex hull contains $\rho[-1,1]^h$ (see \cref{eq:rho-containment}).

\paragraph{Parameter calibration.}
Fix a target number of actions $N$, and write $\ep_N:=\ep(N)$.
Thus the exponent in the Lipton--Markakis--Mehta upper bound is $O((\log_2 N)/\ep_N^2)$,
while the hypothetical running time ruled out by \cref{thm:main-ne} has
base-two logarithm $(\log_2 N)\left(\frac{\log_2 N}{\ep_N^2}\right)^{1-c}$. We choose the reduction parameters in
two regimes.

First suppose that $\log((\log_2 N)/\ep_N^2)$ is at most a sufficiently small constant
multiple of $\log_2 N$. Up to sufficiently small absolute factors, given an $m$-dimensional instance of \BiLinVI{}, we choose $N$ and $t$ so that 
\begin{align}
m\asymp\frac{(\log_2 N)^2}{\ep_N^2\log^2((\log_2 N)/\ep_N^2)},\qquad
t\asymp\ep_N\sqrt m\asymp\frac{\log_2 N}{\log((\log_2 N)/\ep_N^2)}.
\end{align}
Since $\log m=\Theta(\log((\log_2 N)/\ep_N^2))$, the action bound
$2^{O(t\log m)}$ is at most $N$, while the decoded \BiLinVI{} error is
$O(\ep_N\sqrt m/t)=O(1)$. Moreover,
\begin{align}
(\log_2 N)\left(\frac{\log_2 N}{\ep_N^2}\right)^{1-c}
=o\left(\frac{m}{\operatorname{polylog}(m)}\right),
\end{align}
because every fixed power $((\log_2 N)/\ep_N^2)^c$ dominates $\operatorname{polylog}((\log_2 N)/\ep_N^2)$.
Hence the assumed Nash-equilibrium algorithm would solve the hard source
problem in time $2^{\widetilde o(m)}$.

For the regime in which $\log((\log_2 N)/\ep_N^2)$ is a sufficiently large
constant fraction of $\log_2 N$, or equivalently the target accuracy decays at
least polynomially fast in $N$, we take $t=1$ and
$m=\Theta(1/\ep_N^2)$. To obtain tight bounds we need to use a slightly different construction than the one described above; we refer the reader to \cref{sec:nf-game-con} for details.

\paragraph{Extending the technique to free games.} The proof of \cref{thm:freegame-main} follows similar lines to that of \cref{thm:main-ne}. In particular, the standard ``birthday repetition'' technique when applied to free games \cite{AIM14} starts out with a bipartite 2-CSP where both sides $\VL, \VR$ have $m$ vertices (see \cref{lem:freegame-source-normalization}). We consider this CSP as a bipartite graph where there is an edge $(i,j) \in \VL \times \VR$ if there is some constraint on the tuple $(i,j)$.  It then takes the sets of questions to \Lp, \Rp to be the sets of blocks $I \subset \VL$, $J \subset \VR$, respectively, of some size $k < m$ (e.g., to obtain a lower bound of $n^{\tilde\Omega(\log(n)/\ep)}$ one takes $k \asymp \sqrt{m\ep}$). The provers' answers consist of assignments $\phi_{\Lp} : I \to \Sigma$, $\phi_{\Rp} : J \to \Sigma$ of a value (in the alphabet $\Sigma$ of the CSP instance) to each element in the block. The verifier $V$ then accepts if and only if for each $(i,j)$ with $i \in I, j \in J$ for which $e = (i,j)$ forms a constraint in the 2-CSP instance, the provers' answers satisfy the constraint. 

To prove \cref{thm:freegame-main}, we again start with a bipartite 2-CSP but now introduce a free game where the players have asymmetric roles. In particular, as we did for Nash equilibrium, for parameters $t,h$ satisfying $m \leq t\cdot h < 2m$, we partition $\VL$ into $t$ blocks $I \subset \VL$ of size roughly $h$ for the \Lp-player and $\VR$ into $h$ blocks $J \subset [m]$ of size roughly $t$ for the \Rp-player. We consider a free game instance where \Rp-player's questions are blocks $J \subset [m]$, and their answers consist of (a) an assignment $\phi_{\Rp} : J \to \Sigma$ mapping each element of $J$ to a value in $\Sigma$, and (b)  an assignment $\psi_{\Rp}$, which maps each tuple $(j,i)$ where $j \in J$ and $i$ is a neighbor of $j$, to some value $\psi_{\Rp}((j,i)) \in \Sigma$. 

Given that \Rp-player specifies assignments to nodes on both the left and right sides, the role of the \Lp-player is to act as a sort of ``verifier'' (together with the actual verification function $V$ of the free game). In particular, questions to the \Lp-player take the form of blocks $I \subset [m]$, to which the \Lp-player responds with a choice of vector $v \in \MV(h,k) \subseteq \{-1,1\}^h$, to be interpreted as assigning a bit to each element of $J$ (here $\MV(h,k)$ satisfies the analogous properties to $\MV(h,t)$ in \cref{eq:v-prop-intro}, for an appropriate value of $k \approx \ep^2 h$, where $\ep$ is the desired accuracy level). At a high level, the verification function $V$ of the free game evaluates the inner product between $v$ and a hash of the \Rp-player's assignment at each position $i \in I$  (assuming they are valid; see \cref{eq:freegame-raw-verifier}). The hash function is also provided as part of the question to \Lp. 

If the 2-CSP instance has a satisfying assignment, then for arbitrary hash functions, $v$ can be chosen (based on the ground-truth satisfying assignment and the hash) so that this inner product is large. If the 2-CSP is far from having a satisfying assignment, then for any choice of $v$ (even depending on the hash function), the inner product will be small in expectation over the hash function.

At a conceptual level, one key difference between the free game setting and the approximate Nash setting arises because we can not utilize secondary actions to force the \Lp player to play ``valid'' mixtures over $v \in \MV(h,t)$, i.e., those belonging to $\rho \cdot [-1,1]^h$. To circumvent this issue, we add another component to the question provided to \Lp. This additional component takes the form of a $p$-wise independent hash function mapping the block $I$ to $\{-1,1\}$, where $p = \log|\MV(h,k)| \approx k$. This hash function is incorporated into the inner product defining the verification function $V$ mentioned above. Via standard moment bounds for $p$-wise independent random variables (see \cref{lem:freegame-adaptive-selector}), we can bound the amount by which the \Lp-player can ``cheat'' by playing an arbitrary $v \in \MV$. %
 We refer to \cref{sec:free-game} for full details. 

\section{Technical lemmas}
Below we state a couple of technical lemmas which are used throughout the paper. 
\begin{lemma}[Embedding a scaled hypercube in a small simplex]
    \label{lem:scaled-hypercube}
    Fix integers $1 \leq k \leq r$ where $r$ is a power of 2. Then there is a multiset $\MV(r,k) \subseteq \{-1,1\}^r$ of size $r^4 \cdot 2^k$ so that
    \begin{align}
[-1,1]^r \subseteq 4 \sqrt{r/k} \cdot \mathrm{conv}(\MV(r,k)) \nonumber.
    \end{align}
    Further, identifying the elements of $\MV(r,k)$ with the set $[r^42^k]$, each element of $\MV(r,k)$ may be evaluated in time $\mathrm{poly}(r)$ given its index in $[r^42^k]$. 
\end{lemma}
In the case $k=1$, \cref{lem:scaled-hypercube} may be established by taking $\MV = \MV(r,1)$ to be the set of rows of the symmetric Hadamard matrix $H_r \in \{-1,1\}^{r \times r}$ and their negations, repeated as needed to obtain the stipulated multiset size. %
Then we have that $[-1,1]^r \subseteq \sqrt{r} \cdot \mathrm{conv}(\MV)$: indeed, for any $\bx \in [-1,1]^r$, $\ba := \frac{H_r^\top \cdot \bx}{r \sqrt{r}}$ satisfies $\| \ba \|_1 \leq \sqrt{r} \| \ba\|_2 = \| \bx\|_2/\sqrt{r} \leq 1$, and the fact that $H_r \cdot \ba = \bx/\sqrt{r}$ implies that $\bx \in \sqrt{r} \cdot \mathrm{conv}(\MV)$.

\begin{lemma}
    \label{lem:partitions}
    Let $G = (V, E)$ be a bipartite graph with bipartition $V = \VL \cup \VR$ satisfying $|\VL| = |\VR| = m$, and maximum degree $D$. Fix positive integers $t, h$ with $m \leq t\cdot h \leq 2m$. Then $\VL$ can be divided into $t$ labeled, possibly empty blocks $\VL\^1, \ldots, \VL\^t$ of size at most $h$, and $\VR$ can be divided into $h$ labeled, possibly empty blocks $\VR\^1, \ldots, \VR\^{h}$ of size at most $t$ so that, for each pair $i,j$, there are at most $2D(D^2+3)$ edges between $\VL\^i$ and $\VR\^j$. The partitions may be computed deterministically in time $\mathrm{poly}(m)$.
\end{lemma}

\section{A tight lower bound for approximate Nash equilibrium}
\label{sec:nf-game-con}
In this section we prove \cref{thm:main-ne}. 
Consider a \BiLinVI{} instance specified by a tuple $(\AL, \AR, \aL, \aR)$ for which $\AL, \AR$ have at most 3 nonzero entries in every row and column. Define the bipartite graph $G = (V,E)$ with $V = \VL \cup \VR$ and $\VL = \VR = [m]$, and with $E$ being given by the set of $(i,j)$ for which any of $(\AL)_{ij}, (\AL)_{ji}, (\AR)_{ij}, (\AR)_{ji}$ is nonzero. Then $G$ has maximal degree $D \leq 12$. We write $C_D:=2D(D^2+3)$ for the block-pair edge bound in \cref{lem:partitions}.

Fix an integer $1\leq t\leq\sqrt m$, and set
\begin{align}
h = 2^{\lceil \log_2(m/t)\rceil}, \qquad K = 10^7, \qquad \rho = \frac{1}{4}\sqrt{\frac{t}{h}} \label{eq:set-params},
\end{align}
so that $m \leq th < 2m$, $t\leq h$, $h$ is a power of $2$, and $\rho=\Theta(t/\sqrt m)$. We consider the partitions $\VL = \VL\^1 \cup \cdots \cup \VL^t$, $\VR = \VR\^1 \cup \cdots \cup \VR^h$ per \cref{lem:partitions}. We let $\MV,\MW$ be two labeled copies of $\MV(h,t) \subseteq \{-1,1\}^h$, where $\MV(\cdot,\cdot)$ is as in \cref{lem:scaled-hypercube}. Thus
\begin{align}
\rho[-1,1]^h \subseteq \operatorname{conv}(\MV), \qquad
\rho[-1,1]^h \subseteq \operatorname{conv}(\MW).
\label{eq:rho-containment}
\end{align}

\paragraph{Construction of the actions.} We consider a 2-player normal form game $\MG$ with players \Lp (``left'') and  \Rp (``right'') and action sets $\MAL, \MAR$, defined as follows. First, $\MAL$ contains the following actions, which we break into \emph{primary} actions and \emph{secondary} actions:
\begin{itemize}
\item \textbf{Left primary actions.} For each left-side block $I \in \{ \VL\^1, \ldots, \VL\^t \}$ and $v \in \MV$, we include a \emph{primary action} consisting of the tuple $(I, v)$.
\item \textbf{Left secondary actions.} We include a \emph{secondary action} $w$ for each $w \in \MW$. 
\end{itemize}

\noindent Next, $\MAR$ contains the following actions, which we break into \emph{primary} actions and \emph{secondary} actions: 
\begin{itemize}
\item \textbf{Right primary actions.} For each right-side block $J \in \{ \VR\^1, \ldots, \VR\^h \}$ and $z \in \{-1,1\}^t$, we include a \emph{primary action} consisting of the tuple $(J, z)$.
\item \textbf{Right secondary actions.} We include the following three types of \emph{secondary actions}:
\begin{itemize}
\item \textbf{Type I.} First, we include a secondary action for each mapping $f : [t] \to \{ \perp \} \cup ([h] \times \{-1,1\})$. 
\item \textbf{Type II.} Next, we include a secondary action for each vector $z \in \{-1,1\}^t$.
\item \textbf{Type III.} Finally, we include a dummy secondary action $\perp$. 
\end{itemize}
\end{itemize}
We let $\MAL^1, \MAL^2$ be the sets of left primary and secondary actions, respectively, and $\MAR^1, \MAR^2$ be the sets of right primary and secondary actions, respectively.

\paragraph{Payoffs.} Next we describe the payoffs for the players. For $\aL \in \MAL, \aR \in \MAR$, we define payoffs
\begin{align}
\uL(\aL, \aR) :=& 10^{-3}K \cdot \uL^0(\aL, \aR) + \uL^1(\aL, \aR),\nonumber\\
\uR(\aL, \aR) :=& 10^{-3}K \cdot \uR^0(\aL, \aR) + \uR^1(\aL, \aR),\nonumber
\end{align}
where the payoffs $\uL^0, \uL^1, \uR^0, \uR^1$ will be defined below. 

We begin by defining $\uL^0, \uR^0$. First, we suppose that both players play primary actions. For a left-hand side index $i \in [m]$, we let $\piL(i) \in [h]$ denote the index of $i$ in the unique left-hand side block $I \in \{ \VL\^1, \ldots, \VL\^t \}$ containing $i$. For a right-hand side index $j \in [m]$, let $\piR(j) \in [t]$ denote the index of $j$ in the unique right-hand side block $J \in \{\VR\^1, \ldots, \VR\^h \}$ containing $j$. For a left primary action $(I,v) \in \{ \VL\^1, \ldots, \VL\^t \} \times \MV$ and a right primary action $(J,z) \in \{ \VR\^1, \ldots, \VR\^h \} \times \{-1,1\}^t$, we define the payoffs
\begin{align}
\uL^0((I,v), (J,z)) &:= \frac{1}{K} \cdot \left(\frac{1}{h} \sum_{i \in I} a_i \cdot v_{\piL(i)} + \sum_{i \in I, j \in J} A_{ij} \cdot v_{\piL(i)} \cdot z_{\piR(j)} \right),\label{eq:both-primary-left}\\
\uR^0((I,v), (J,z)) &:=\frac{1}{K} \cdot \left( \frac{\rho}{t} \sum_{j \in J} b_j \cdot z_{\piR(j)} + \sum_{i \in I, j \in J} B_{ji} \cdot v_{\piL(i)} \cdot z_{\piR(j)}\right).\label{eq:both-primary-right}
\end{align}

Next, we describe the payoffs under $\uL^0, \uR^0$ when one player plays a primary action and the other plays a secondary action. First, if the left player plays a primary action $(I,v)$:
\begin{itemize}
\item If the right player plays a secondary action $f : [t] \to \{ \perp \} \cup ([h] \times \{-1,1\})$ of type I, we define the payoffs
\begin{align}
\uR^0((I,v), f) = \begin{cases}
0 &: f(s) = \perp,\\
b \cdot v_j - \rho &: f(s) = (j,b) \in [h] \times \{-1,1\},
\end{cases}\qquad \uL^0((I,v), f) = -\uR^0((I,v), f)\label{eq:typei-def}.
\end{align}
Here $s\in[t]$ is the index for which $I=\VL\^s$.
\item If the right player plays a secondary action $z \in \{-1,1\}^t$ of type II, we define the payoffs
\begin{align}
\uR^0((I,v), z) = z_s - \frac{1}{t} \sum_{s' \in [t]} z_{s'}, \qquad \uL^0((I,v), z) = -\uR^0((I,v), z), \qquad I = \VL\^s, \ \ s \in [t]\label{eq:typeii-def}.
\end{align}
\item If the right player plays the action $\perp$ of type III, we define the payoffs $\uR^0((I,v), \perp) = \uL^0((I,v), \perp) = 0$.
\end{itemize}
Next, if the right player plays a primary action $(J,z)$:
\begin{itemize}
\item If the left player plays a secondary action $w \in \MW$, then we define the payoffs
\begin{align}
\uL^0(w, (J,z)) = w_g - \frac{1}{h} \sum_{g' \in [h]} w_{g'}, \qquad \uR^0(w, (J,z)) = -\uL^0(w, (J,z)), \qquad J = \VR\^g, \ \ g \in [h]\label{eq:typew-def}.
\end{align}
\end{itemize}
If both players play secondary actions $\aL, \aR$, then we define $\uL^0(\aL, \aR) = \uR^0(\aL, \aR) = 0$. 

Finally, we define the payoffs $\uL^1, \uR^1$ to be as follows. We define
\begin{align}
\uL^1(\aL, \aR) = -\uR^1(\aL, \aR) = \begin{cases}
    K &: \aL \in \MAL^1, \aR \in \MAR^1 \text{ or } \aL \in \MAL^2, \aR \in \MAR^2 \\
    -K &: \aL \in \MAL^1, \aR \in \MAR^2 \text{ or } \aL \in \MAL^2, \aR \in \MAR^1.
\end{cases}\nonumber
\end{align}

\subsection{Enforcing equilibrium constraints}
For primary distributions $\pL\in\Delta(\MAL^1)$ and $\pR\in\Delta(\MAR^1)$, let $\barpL\in\Delta([t])$ and $\barpR\in\Delta([h])$ be the induced distributions over blocks: for $s\in[t]$ and $g\in[h]$,
\begin{align}
\barpL(s) = \Pr_{(I,v) \sim \pL}(I = \VL\^s), \qquad \barpR(g) = \Pr_{(J,z) \sim \pR}(J = \VR\^g).\nonumber
\end{align}
For $I=\VL\^s$ and $J=\VR\^g$ which occur under $\pL, \pR$, respectively, with positive probability, define the \emph{conditional means}
\begin{align}
\pL|_s:=\E_{(I',v')\sim\pL}[v'\mid I'=I],
\qquad
\pR|_g:=\E_{(J',z')\sim\pR}[z'\mid J'=J], \label{eq:define-pli}
\end{align}
and set the corresponding mean to zero when the block (i.e., $I$ or $J$) has probability zero under $\pL, \pR$, respectively. We define the following error functions:
\begin{align}
\delta_s(\pL)&:=\left(\|\pL|_s\|_\infty-\rho\right)_+, &
\delta(\pL)&:=\sum_{s\in[t]}\barpL(s)\delta_s(\pL),\nonumber\\
\Delta_{\mathsf L}(\pL)&:=\left\|\barpL-\frac{\mathbf 1}{t}\right\|_1, &
\Delta_{\mathsf R}(\pR)&:=\left\|\barpR-\frac{\mathbf 1}{h}\right\|_1. \label{eq:projection-errors}
\end{align}
\begin{lemma}
    \label{lem:secondary-deviations}
For any distributions $\pL \in \Delta(\MAL^1), \pR \in \Delta(\MAR^1)$, it holds that
\begin{align}
\max_{\actR \in \MAR^2} \uR^0(\pL, \actR) =& \max \left\{\delta(\pL),\Delta_{\mathsf L}(\pL)\right\},\label{eq:left-secondary-deviation}\\
\max_{\actL \in \MAL^2} \uL^0(\actL, \pR) \geq & \rho\Delta_{\mathsf R}(\pR).
\label{eq:right-secondary-deviation}
\end{align}
\end{lemma}
\begin{proof}
For a type-I secondary action $f : [t] \to \{ \perp \} \cup ([h] \times \{-1,1\})$, the values $f(s)$ may be chosen independently across blocks $s \in[t]$. Hence
\begin{align}
\max_f\uR^0(\pL,f)
=\sum_{s\in[t]}\barpL(s)\left(\|\pL|_s\|_\infty-\rho\right)_+
=\delta(\pL).
\end{align}
For a type-II action $z\in\{-1,1\}^t$, the expected payoff to \Rp is
\begin{align}
\left\langle z,\barpL-\frac{\mathbf 1}{t}\right\rangle,
\end{align}
whose maximum over sign vectors $z \in \{-1,1\}^t$ is $\Delta_{\mathsf L}(\pL)$. The type-III secondary action has payoff zero, so taking the best of the three types proves \cref{eq:left-secondary-deviation}.

The payoff to \Lp of a player-\Lp secondary action $w\in\MW$ against $\pR$ is $\langle w,\barpR-\mathbf 1/h\rangle$. Since \cref{eq:rho-containment} gives $\rho[-1,1]^h\subseteq\operatorname{conv}(\MW)$,
\begin{align}
\max_{w\in\MW}\left\langle w,\barpR-\frac{\mathbf 1}{h}\right\rangle
\geq \rho\Delta_{\mathsf R}(\pR),
\end{align}
which proves \cref{eq:right-secondary-deviation}.
\end{proof}

Now consider a product distribution $(\pL, \pR) \in \Delta(\MAL) \times \Delta(\MAR)$, to be interpreted as an $\ep$-Nash equilibrium of $\MG$. We let $\pL^1 \in \Delta(\MAL^1), \pL^2 \in \Delta(\MAL^2)$ denote the conditional distributions of $\actL \sim \pL$ given $\actL \in \MAL^1$ and $\actL \in \MAL^2$, respectively, and define $\pR^1,\pR^2$ analogously. With a slight abuse of notation, we write $\barpL := \barpL^1, \barpR := \barpR^1$ to denote the block marginals and, for $s \in [t], g \in [h]$, $\pL|_s := \pL^1|_s, \pR|_g := \pR^1|_g$ to denote the conditional means. We also write  $\uL^0(\pL, \pR) := \E_{\actL \sim \pL, \actR \sim \pR}[\uL^0(\actL, \actR)]$ and similarly for $\uR^0$. 

\begin{lemma}
\label{lem:conditional-deviations}
Suppose $(\pL, \pR)$ is an $\ep$-Nash equilibrium of $\MG$, with $\ep \leq K/100$. Then, with $\ep' = 16000\ep/K$, the following statements hold:
\begin{enumerate}
\item $\pL(\MAL^1), \pR(\MAR^1) \in [1/4, 3/4]$.
\item For every $\pL' \in \Delta(\MAL^1), \pR' \in \Delta(\MAR^1)$, we have, for some values $\cL, \cR \in [1/3, 3]$ (depending on $\pL, \pR$),
\begin{align}
\uL^0(\pL^1, \pR^1) + \cL \cdot \uL^0(\pL^1, \pR^2) &\geq \uL^0(\pL', \pR^1) + \cL \cdot \uL^0(\pL', \pR^2) - \ep',\label{eq:left-primary-deviation}\\
\uR^0(\pL^1, \pR^1) + \cR \cdot \uR^0(\pL^2, \pR^1) &\geq \uR^0(\pL^1, \pR') + \cR \cdot \uR^0(\pL^2, \pR') - \ep'.\label{eq:right-primary-deviation}
\end{align}
\item It holds that
\begin{align}
\max_{\actR \in \MAR^2} \uR^0(\pL^1, \actR) - \uR^0(\pL^1, \pR^2) \leq \ep', \qquad \max_{\actL \in \MAL^2} \uL^0(\actL, \pR^1) - \uL^0(\pL^2, \pR^1)\leq \ep'.\label{eq:al-ar-deviations}
\end{align}
\end{enumerate}
\end{lemma}
\begin{proof}
Let $\gamma_{\mathsf L}:=\pL(\MAL^1)$ and $\gamma_{\mathsf R}:=\pR(\MAR^1)$. By \cref{lem:partitions} and $D\leq12$, for any primary actions $\actL \in \MAL^1, \actR \in \MAR^1$, we have  $|K\uL^0(\actL, \actR)|,|K\uR^0(\actL, \actR)|\leq3529$; further, for all $\actL, \actR$ which are not both primary, we have $|K \uL^0(\actL, \actR)|, |K\uR^0(\actL, \actR)| \leq 2K$. Thus, by our choice of $K$ in \cref{eq:set-params}, changing one player's action changes $10^{-3}K \cdot \uL^0, 10^{-3}K \cdot \uR^0$ by at most $4\cdot10^{-3}K$.

Suppose $\gamma_{\mathsf R}\geq3/4$. Then by definition of $\uL^1$, for every left primary action $\actL$ and left secondary action $\actL'$, we have $\uL(\actL,\pR) \geq (1-4\cdot10^{-3})K +  \uL(\actL', \pR)$. %
Let $\actL^\star \in \MAL^1$ be an arbitrary left primary action. The $\ep$-Nash condition therefore gives
\begin{align}
\ep
&\geq \E_{\actL\sim\pL}\left[\uL(\actL^\star,\pR)-\uL(\pL,\pR)\right]\geq(1-\gamma_{\mathsf L})(1-4\cdot10^{-3})K,\nonumber
\end{align}
and hence $1-\gamma_{\mathsf L}\leq\ep/((1-4\cdot10^{-3})K)<1/90$, where we have used $\ep\leq K/100$.
At this value of $\gamma_{\mathsf L}$, every right secondary action is better than every right primary action by more than $1.9K$. Moving the right player's primary mass, which is at least $3/4$, to an arbitrary  secondary action, is therefore profitable by more than $K$, a contradiction. The case $\gamma_{\mathsf R}\leq1/4$ is identical with both roles reversed: it first gives $\gamma_{\mathsf L}<1/90$, after which moving the right player's secondary-role mass to the primary role gains more than $K$. Hence $\gamma_{\mathsf R}\in[1/4,3/4]$, and the same argument with the players and roles interchanged gives $\gamma_{\mathsf L}\in[1/4,3/4]$.

Set $\cL=(1-\gamma_{\mathsf R})/\gamma_{\mathsf R}$ and $\cR=(1-\gamma_{\mathsf L})/\gamma_{\mathsf L}$, so $\cL,\cR\in[1/3,3]$. Fix any $\pL' \in \MAL^1$. Let $\tilpL$ be the following modified version of $\pL$: in particular, change the distribution of $\actL \sim \tilpL$ conditioned on $\actL \in \MAL^1$ to be $\pL'$ (instead of $\pL^1$). Deviating from $\pL$ to $\tilpL$ gives the $\mathsf{L}$-player a gain of
\begin{align}
10^{-3}K\gamma_{\mathsf L}\gamma_{\mathsf R}\Big(&\uL^0(\pL',\pR^1)+\cL\uL^0(\pL',\pR^2)
-\uL^0(\pL^1,\pR^1)-\cL\uL^0(\pL^1,\pR^2)\Big).\label{eq:cl-cr-proof}
\end{align}
Since $\gamma_{\mathsf L}\gamma_{\mathsf R}\geq1/16$, the $\ep$-Nash condition gives \cref{eq:left-primary-deviation}; the analogous computation for player-$\mathsf{R}$ gives \cref{eq:right-primary-deviation}.

It remains to prove \cref{eq:al-ar-deviations}. Choose
$\actL^\star\in\argmax_{\actL\in\MAL^2}\uL^0(\actL,\pR^1)$, and let
$\tilpL$ be the following modified version of $\pL$: change the distribution of $\actL \sim \tilpL$ conditioned on $\actL \in \MAL^2$ to be $\actL^\star$ (instead of $\pL^2$). This
deviation preserves the $\mathsf{L}$-player's probabilities of playing an action in $\MAL^1$ and $\MAL^2$; further, the  player receives $0$ payoff under $\uL^0$ when both play secondary actions, so the gain in its payoff is
\begin{align}
10^{-3}K(1-\gamma_{\mathsf L})\gamma_{\mathsf R}
\left(
\max_{\actL\in\MAL^2}\uL^0(\actL,\pR^1)
-\uL^0(\pL^2,\pR^1)
\right)\leq \ep.
\end{align}
Since $(1-\gamma_{\mathsf{L}})\gamma_{\mathsf{R}} \geq 1/16$, \cref{eq:al-ar-deviations} for $\mathsf{L}$ follows; the analogous inequality for player-$\mathsf{R}$ follows symmetrically.
\end{proof}

Next, given distributions $\pL \in \Delta(\MAL), \pR \in \Delta(\MAR)$, we define ``projected'' variants, as below:
\begin{definition}[Projected primary distributions]
    \label{def:projected}
Given $\pL \in \Delta(\MAL), \pR \in \Delta(\MAR)$, we define their \emph{projected} versions $\hatpL \in \Delta(\MAL^1), \hatpR \in \Delta(\MAR^1)$ as follows:
\begin{itemize}
\item $\hatpL$ is a distribution over $(I,v)$ whose marginal over the block index $I$ is uniform and whose conditional mean of $v$ is a clipped version of $\pL|_s$ for each $s \in [t]$:
\begin{align}
\Pr_{(I,v) \sim \hatpL}(I = \VL\^s) =& \frac{1}{t} \qquad \forall s \in [t]\label{eq:block-marginal-uniform}\\
\E_{(I,v)\sim\hatpL}[v\mid I=\VL\^s]
=&\operatorname{clip}(\pL|_s,[-\rho,\rho]^h), \qquad s\in[t], \label{eq:projected-left}
\end{align}
Such a distribution exists since $\rho [-1,1]^h \subseteq \mathrm{conv}(\MV)$, i.e., \cref{eq:rho-containment}; if such a distribution is not unique, we let $\hatpL$ be an arbitrary choice of one.
\item $\hatpR$ is the distribution over $(J,z)$ whose marginal over the block index $J$ is uniform and whose conditional mean of $z$ is $\pR|_g$ for each $g \in [h]$:
\begin{align}
\Pr_{(J,z) \sim \hatpR}(J = \VR\^g) = \frac{1}{h} \qquad \forall g \in [h] \label{eq:block-marginal-uniform-right}\\
\E_{(J,z) \sim \hatpR}[z \mid J = \VR\^g] = \pR|_g \qquad \forall g \in [h] \label{eq:projected-right}.
\end{align}
The distribution $\hatpR$ exists because $[-1,1]^t=\operatorname{conv}(\{-1,1\}^t)$.
\end{itemize}
\end{definition} %

The next lemma bounds the differences in utility by passing to the projected versions of $\pL, \pR$.
\begin{lemma}[Projection bounds]
\label{lem:projection-bounds}
Fix distributions $\pL \in \Delta(\MAL), \pR \in \Delta(\MAR)$, and let $\hatpL, \hatpR$, respectively, be their projected versions via \cref{def:projected}. Then
\begin{align}
K\big(&|\uL^0(\pL^1,\pR^1)-\uL^0(\hatpL,\pR^1)|
+|\uR^0(\pL^1,\pR^1)-\uR^0(\hatpL,\pR^1)|\big)
&\leq 10(C_D+1)\bigl(\delta(\pL^1)+\rho\Delta_{\mathsf L}(\pL^1)\bigr), \label{eq:left-projection-bound}\\
K\big(&|\uL^0(\pL^1,\pR^1)-\uL^0(\pL^1,\hatpR)|
+|\uR^0(\pL^1,\pR^1)-\uR^0(\pL^1,\hatpR)|\big)
&\leq 10(C_D+1)\bigl(\rho+\delta(\pL^1)\bigr)\Delta_{\mathsf R}(\pR^1). \label{eq:right-projection-bound}
\end{align}
Moreover,
\begin{align}
\uR^0(\hatpL,\actR)\leq0 \quad\text{for every $\actR\in\MAR^2$},
\qquad
\uL^0(\actL,\hatpR)=0 \quad\text{for every $\actL\in\MAL^2$}. \label{eq:projected-tests}
\end{align}
\end{lemma}
\begin{proof}
We first prove \cref{eq:left-projection-bound}. Clipping the conditional mean
$\pL|_s$ to $[-\rho,\rho]^h$ changes each of its coordinates by at most
$\delta_s(\pL^1)$. Let $\tilpL \in \Delta(\MAL^1)$ be a modification of $\pL^1$ where each of the conditional means is the clipped version $\mathrm{clip}(\pL|_s, [-\rho,\rho]^h)$ (i.e., per \cref{eq:projected-left}). Using the fact that there are at most $C_D$ edges between each pair $\VL\^i, \VR\^j$ (per \cref{lem:partitions}), it follows from the definition of $\uL^0,\uR^0$ in \cref{eq:both-primary-left,eq:both-primary-right} that 
\begin{align}
K\left( | \uL^0(\pL^1, \pR^1) - \uL^0(\tilpL, \pR^1)| + |\uR^0(\pL^1, \pR^1) - \uR^0(\tilpL, \pR^1)|\right) \leq (2C_D+1) \cdot \delta(\pL^1)\label{eq:first-projection-proof}.
\end{align}
Since each entry of $\tilpL|_s$ is bounded in absolute value by $\rho$ (for each $s \in [t]$), it again follows from the definition of $\uL^0,\uR^0$ that %
\begin{align}
K\left( | \uL^0(\hatpL, \pR^1) - \uL^0(\tilpL, \pR^1)| + |\uR^0(\hatpL, \pR^1) - \uR^0(\tilpL, \pR^1)|\right) \leq (2C_D+1) \cdot \rho \cdot \left\| \barpL - \mathbf{1}/t \right\|_1 = (2C_D+1) \cdot \rho \cdot \Delta_{\mathsf{L}}(\pL^1)\nonumber,
\end{align}
which gives \cref{eq:left-projection-bound} when combining with \cref{eq:first-projection-proof}.

To establish \cref{eq:right-projection-bound}, we note that for each $s \in [t]$, $\| \pL|_s \|_\infty \leq \rho + \delta_s(\pL)$, which gives
\begin{align*}
K\Big(|\uL^0(\pL,\pR^1)
-\uL^0(\pL,\hatpR)|
+|\uR^0(\pL,\pR^1)
-\uR^0(\pL,\hatpR)|\Big) \leq & 
(2C_D+1)\bigl(\rho+\delta(\pL)\bigr)
\Delta_{\mathsf R}(\pR^1)\\
\leq& 10(C_D+1)\bigl(\rho+\delta(\pL)\bigr)
\Delta_{\mathsf R}(\pR^1).
\end{align*}

It remains to prove \cref{eq:projected-tests}. Against $\hatpL$, the payoff (under $\uR^0$) of
a $\mathsf{R}$-player type-I secondary action (see \cref{eq:typei-def}) is at most $0$ since for each $s \in [t]$, for any $g \in [h], b \in \{-1,1\}$, we have
\begin{align*}
    b \cdot (\hatpL|_s)_g - \rho \leq \| \hatpL|_s \|_\infty - \rho \leq 0.
\end{align*}
The payoff of a type-II action $z\in\{-1,1\}^t$ against $\hatpL$ is zero because the left block
marginal is uniform (see \cref{eq:typeii-def}):
\begin{align*}
\uR^0(\hatpL,z)
=\frac1t\sum_{s\in[t]}z_s-\frac1t\sum_{s\in[t]}z_s=0.
\end{align*}
A type-III action also has payoff zero. Finally, for every
$w\in\MW=\MAL^2$, the uniform right block marginal gives
\begin{align*}
\uL^0(w,\hatpR)
=\frac1h\sum_{g\in[h]}w_g-\frac1h\sum_{g\in[h]}w_g=0.
\end{align*}
This proves \cref{eq:projected-tests}.
\end{proof}

\subsection{Completing the proof of the reduction}
Let $(\pL, \pR) \in \Delta(\MAL) \times \Delta(\MAR)$ be an $\ep$-Nash equilibrium of the game $\MG$ defined above for the \BiLinVI{} instance $(\AL, \AR,\aL,\aR)$. In \cref{lem:derive-linvi-solution} below, we show how $(\pL, \pR)$ may efficiently be converted into a solution to the original \BiLinVI{} instance $(\AL, \AR, \aL, \aR)$. To do so, we define $\bx, \by \in [-1,1]^m$ depending on $\pL, \pR$, as follows. For each $i \in [m]$, choose the unique $s\in[t]$ for which $i\in\VL\^s$ and set
\begin{align}
\bx_i = \mathrm{clip}\left(\frac{(\pL|_s)(\piL(i))}{\rho}, [-1,1] \right). \label{eq:xi-definition}
\end{align}
Here recall that $\piL(i)$ denotes the index of $i$ in $\VL\^s$.

In a similar manner, for each $j \in [m]$, choose the unique $g\in[h]$ for which $j\in\VR\^g$ and set
\begin{align}
\by_j = (\pR|_g)(\piR(j)).\label{eq:yj-definition}
\end{align}

\begin{lemma}
\label{lem:derive-linvi-solution}
There is a sufficiently small absolute constant $\ep_0 > 0$, so that for any $0 < \ep \leq \ep_0$, if we set $t = \lceil \ep \sqrt{m/\ep_0} \rceil$, then the following holds. 
Suppose that $(\pL, \pR)$ is an $\ep$-Nash equilibrium of the game $\MG$. Then $(\bx, \by)$ as defined in \cref{eq:xi-definition,eq:yj-definition} is an $80000 \cdot \sqrt{32 \ep_0}$%
-approximate solution to the \BiLinVI{} instance $(\AL, \AR, \aL, \aR)$. 
\end{lemma}
\begin{proof}
Set $\ep'=16000\ep/K$. We may assume that $\ep'\leq\rho/13$: otherwise, since every row of $\AL,\AR$ has at most three nonzero entries (and the entries are bounded in $[-1,1]$), each of the two \BiLinVI{} regrets is at most $4\leq 5K\ep'/\rho$, and the claim is immediate. This assumption also implies $\ep\leq K/100$, so \cref{lem:conditional-deviations} applies; further, let $\cL, \cR \in [1/3,3]$ be the constants from \cref{lem:conditional-deviations} given $\pL, \pR$. By \cref{lem:secondary-deviations},
\begin{align}
\Gamma:=\max_{\actR\in\MAR^2}\uR^0(\pL^1,\actR)
=\max\{\delta(\pL^1),\Delta_{\mathsf L}(\pL^1)\},
\qquad
\max_{\actL\in\MAL^2}\uL^0(\actL,\pR^1)\geq\rho\Delta_{\mathsf R}(\pR^1).
\end{align}
Thus \cref{eq:al-ar-deviations} gives
\begin{align}
\uR^0(\pL^1,\pR^2)\geq \Gamma-\ep', \qquad
\uL^0(\pL^2,\pR^1)\geq \rho\cdot \Delta_{\mathsf R}(\pR^1)-\ep'. \label{eq:test-lower-bounds}
\end{align}

Let $\hatpL, \hatpR$ be the projected distributions defined from $\pL, \pR$ as in \cref{def:projected}. Use $\pL'=\hatpL$ in \cref{eq:left-primary-deviation}. Since $\pR^2$ is supported on $\mathsf{R}$-player secondary actions, \cref{eq:projected-tests} gives $\uR^0(\hatpL,\pR^2)\leq0$. We therefore obtain
\begin{align}
\cL(\Gamma-\ep')
&\leq \cL\bigl(\uR^0(\pL^1,\pR^2)-\uR^0(\hatpL,\pR^2)\bigr)
&&\text{\scriptsize by \cref{eq:test-lower-bounds,eq:projected-tests}}\nonumber\\
&=\cL\bigl(\uL^0(\hatpL,\pR^2)-\uL^0(\pL^1,\pR^2)\bigr)
&&\text{\scriptsize since $\uL^0=-\uR^0$ on $(\pL^1, \pR^2)$ and $(\hatpL, \pR^2)$}\nonumber\\
&\leq \uL^0(\pL^1,\pR^1)-\uL^0(\hatpL,\pR^1)+\ep'
&&\text{\scriptsize by \cref{eq:left-primary-deviation}}\nonumber\\
&\leq \frac{10(C_D+1)}{K}\bigl(\delta(\pL^1)+\rho\cdot\Delta_{\mathsf L}(\pL^1)\bigr)+\ep'
&&\text{\scriptsize by \cref{eq:left-projection-bound}}\nonumber\\
&\leq \frac{20(C_D+1)}{K}\Gamma+\ep'
&&\text{\scriptsize since $\delta(\pL^1),\Delta_{\mathsf L}(\pL^1)\leq \Gamma$ and $\rho\leq1$}\nonumber.
\end{align}
Multiplying by $K$ and rearranging, we obtain
\begin{align}
\bigl(\cL K-20(C_D+1)\bigr)\Gamma
&\leq(\cL+1)K\ep'\nonumber.
\end{align}
Since $\cL\geq1/3$, $\cL+1\leq4$, and $10(C_D+1)\leq K/100$, we obtain
\begin{align}
\delta(\pL^1)\leq \Gamma\leq13\ep', \qquad
\Delta_{\mathsf L}(\pL^1)\leq \Gamma\leq13\ep'. \label{eq:left-errors-final}
\end{align}

Now use $\pR'=\hatpR$ in \cref{eq:right-primary-deviation}. Since $\pL^2$ is supported on left secondary actions, \cref{eq:projected-tests} gives $\uL^0(\pL^2,\hatpR)=0$. We therefore obtain
\begin{align}
\cR\bigl(\rho\Delta_{\mathsf R}(\pR^1)-\ep'\bigr)
&\leq \cR\bigl(\uL^0(\pL^2,\pR^1)-\uL^0(\pL^2,\hatpR)\bigr)
&&\text{\scriptsize by \cref{eq:test-lower-bounds,eq:projected-tests}}\nonumber\\
&=\cR\bigl(\uR^0(\pL^2,\hatpR)-\uR^0(\pL^2,\pR^1)\bigr)
&&\text{\scriptsize since $\uR^0=-\uL^0$ on $(\pL^2, \pR^1)$ and $(\pL^2, \hatpR)$}\nonumber\\
&\leq \uR^0(\pL^1,\pR^1)-\uR^0(\pL^1,\hatpR)+\ep'
&&\text{\scriptsize by \cref{eq:right-primary-deviation}}\nonumber\\
&\leq \frac{10(C_D+1)}{K}\bigl(\rho+\delta(\pL^1)\bigr)\Delta_{\mathsf R}(\pR^1)+\ep'
&&\text{\scriptsize by \cref{eq:right-projection-bound}}\nonumber
\end{align}
Rearranging then gives
\begin{align}
\left(\cR K\rho-10(C_D+1)(\rho+\delta(\pL^1))\right)\Delta_{\mathsf R}(\pR^1)
&\leq(\cR+1)K\ep'
\end{align}
By the assumption $\ep'\leq\rho/13$ and \cref{eq:left-errors-final}, we have $\delta(\pL^1)\leq\rho$. Since $\cR\geq1/3$ and $10(C_D+1)\leq K/100$, the coefficient on the left is at least $K\rho/4$, and hence
\begin{align}
\Delta_{\mathsf R}(\pR^1)\leq\frac{16\ep'}{\rho}. \label{eq:right-errors-final}
\end{align}

Fix $\bx',\by'\in[-1,1]^m$. Now, choose some $\pL^{\bx'}\in\Delta(\MAL^1)$ over $(I,v) \in \MAL^1$ whose marginal over the block index $I$ is uniform and whose conditional distributions of $v$ are given by $\bx'$: formally, for $s \in [t], i \in \VL\^s$, 
\begin{align*}
\Pr_{(I,v)\sim\pL^{\bx'}}(I=\VL\^s)&=\frac1t, \qquad
\E_{(I,v)\sim\pL^{\bx'}}[v_{\piL(i)}\mid I=\VL\^s]=\rho\cdot \bx_i'.
\end{align*}
Such a distribution exists because $\rho[-1,1]^h\subseteq\operatorname{conv}(\MV)$ by \cref{eq:rho-containment}. Similarly, choose any $\pR^{\by'}\in\Delta(\MAR^1)$ satisfying, for all $g \in [h], j \in \VR\^g$,
\begin{align*}
\Pr_{(J,z)\sim\pR^{\by'}}(J=\VR\^g)&=\frac1h, \qquad
\E_{(J,z)\sim\pR^{\by'}}[z_{\piR(j)}\mid J=\VR\^g]=\by_j'.
\end{align*}
This distribution exists because $[-1,1]^t=\operatorname{conv}(\{-1,1\}^t)$. By construction, every right secondary action has nonpositive payoff (under $\uR^0$) against $\pL^{\bx'}$, whereas every left secondary action has zero payoff (under $\uL^0$) against $\pR^{\by'}$. Hence
$\uR^0(\pL^{\bx'},\pR^2)\leq0$ and $\uL^0(\pL^2,\pR^{\by'})=0$. We may now compute:
\begin{align}
K\cdot \left(\uL^0(\pL^{\bx'},\pR^1)-\uL^0(\pL^1,\pR^1)\right)
&\leq K\cL\left(\uL^0(\pL^1,\pR^2)-\uL^0(\pL^{\bx'},\pR^2)\right)+K\ep'
&&\text{\scriptsize by \cref{eq:left-primary-deviation}}\nonumber\\
&=K\cL\left(-\uR^0(\pL^1,\pR^2)+\uR^0(\pL^{\bx'},\pR^2)\right)+K\ep'
&&\text{\scriptsize since $\uL^0=-\uR^0$ for these inputs}\nonumber\\
&\leq K\cL(-\Gamma+\ep')+K\ep'
&&\text{\scriptsize by \cref{eq:test-lower-bounds} and $\uR^0(\pL^{\bx'},\pR^2)\leq0$}\nonumber\\
&\leq4K\ep'
&&\text{\scriptsize since $\Gamma\geq0$ and $\cL\leq3$},\\
K\cdot \left(\uR^0(\pL^1,\pR^{\by'})-\uR^0(\pL^1,\pR^1)\right)
&\leq K\cR\left(\uR^0(\pL^2,\pR^1)-\uR^0(\pL^2,\pR^{\by'})\right)+K\ep'
&&\text{\scriptsize by \cref{eq:right-primary-deviation}}\nonumber\\
&=K\cR\left(-\uL^0(\pL^2,\pR^1)+\uL^0(\pL^2,\pR^{\by'})\right)+K\ep'
&&\text{\scriptsize since $\uR^0=-\uL^0$ for these inputs}\nonumber\\
&\leq K\cR\left(-\rho\Delta_{\mathsf R}(\pR^1)+\ep'\right)+K\ep'
&&\text{\scriptsize by \cref{eq:test-lower-bounds} and $\uL^0(\pL^2,\pR^{\by'})=0$}\nonumber\\
&\leq4K\ep'
&&\text{\scriptsize since $\Delta_{\mathsf R}(\pR^1)\geq0$ and $\cR\leq3$}.
\end{align}

Using \cref{lem:projection-bounds} and the above displays, as well as \cref{eq:left-errors-final,eq:right-errors-final} and $10(C_D+1)\leq K/100$, we obtain
\begin{align}
K\left(\uL^0(\pL^{\bx'},\hatpR)-\uL^0(\hatpL,\hatpR)\right)
&\leq4K\ep'+10(C_D+1)\bigl(\delta(\pL^1)+\rho\Delta_{\mathsf L}(\pL^1)\bigr)\nonumber\\
&\qquad+20(C_D+1)\rho\Delta_{\mathsf R}(\pR^1)
\leq5K\ep',\label{eq:left-projected-deviation}\\
K\left(\uR^0(\hatpL,\pR^{\by'})-\uR^0(\hatpL,\hatpR)\right)
&\leq4K\ep'+20(C_D+1)\bigl(\delta(\pL^1)+\rho\Delta_{\mathsf L}(\pL^1)\bigr)\nonumber\\
&\qquad+10(C_D+1)\rho\Delta_{\mathsf R}(\pR^1)
\leq5K\ep'.\label{eq:right-projected-deviation}
\end{align}
Next, the definition of $\uL^0$ when both players play a primary action (namely, \cref{eq:both-primary-left}) gives
\begin{align*}
&K\left(\uL^0(\pL^{\bx'},\hatpR)-\uL^0(\hatpL,\hatpR)\right)\\
&\quad=\frac{1}{th}\sum_{s\in[t]}\sum_{g\in[h]}
\left(\frac{\rho}{h}\sum_{i\in\VL\^s}a_i(\bx_i'-\bx_i)
+\rho\sum_{\substack{i\in\VL\^s\\j\in\VR\^g}}A_{ij}(\bx_i'-\bx_i)\by_j\right)
&&\text{\scriptsize by the definition of $\uL^0$}\\
&\quad=\frac{\rho}{th}\left(
\sum_{i\in[m]}a_i(\bx_i'-\bx_i)
+\sum_{i,j\in[m]}A_{ij}(\bx_i'-\bx_i)\by_j\right)
&&\text{\scriptsize since the blocks partition $[m]$}\\
&\quad=\frac{\rho}{th}\left\langle\bx'-\bx,\AL\by+\aL\right\rangle.
\end{align*}
Similarly,
\begin{align*}
&K\left(\uR^0(\hatpL,\pR^{\by'})-\uR^0(\hatpL,\hatpR)\right)\\
&\quad=\frac{1}{th}\sum_{s\in[t]}\sum_{g\in[h]}
\left(\frac{\rho}{t}\sum_{j\in\VR\^g}b_j(\by_j'-\by_j)
+\rho\sum_{\substack{i\in\VL\^s\\j\in\VR\^g}}B_{ji}\bx_i(\by_j'-\by_j)\right)
&&\text{\scriptsize by the definition of $\uR^0$}\\
&\quad=\frac{\rho}{th}\left(
\sum_{j\in[m]}b_j(\by_j'-\by_j)
+\sum_{i,j\in[m]}B_{ji}\bx_i(\by_j'-\by_j)\right)
&&\text{\scriptsize since the blocks partition $[m]$}\\
&\quad=\frac{\rho}{th}\left\langle\by'-\by,\AR\bx+\aR\right\rangle.
\end{align*}
Thus \cref{eq:left-projected-deviation,eq:right-projected-deviation} become
\begin{align}
\frac{\rho}{th}\left\langle\bx'-\bx,\AL\by+\aL\right\rangle\leq5K\ep',\qquad
\frac{\rho}{th}\left\langle\by'-\by,\AR\bx+\aR\right\rangle\leq5K\ep'.
\end{align}
Finally, $th\leq2m$. Dividing by $2m$ and maximizing over $\bx',\by'$ shows that both \BiLinVI{} regrets are at most
\begin{align}
\frac{5K\ep' th}{2m\rho}\leq\frac{5K\ep'}{\rho}=\frac{80000\ep}{\rho} \leq 80000 \cdot \sqrt{32\ep_0} \nonumber, %
\end{align}
as desired.
\end{proof}

\begin{proof}[Proof of \cref{thm:main-ne}]
Fix an $m$-dimensional \BiLinVI{} instance satisfying the sparsity promise
in \cref{lem:bilinvi-hard}. Let $\eta>0$ be the absolute constant in the hardness result for
\BiLinVI{} from \cref{lem:bilinvi-hard}, and decrease the constant $\ep_0$
in \cref{lem:derive-linvi-solution}, if necessary, so that
$80000\sqrt{32\ep_0}\leq\eta$. Fix $c>0$ and the function $\ep(\cdot)$
from the theorem, and suppose toward a contradiction that the asserted
algorithm exists. Replacing $c$ by $\min\{c,1/2\}$ only weakens its
running-time guarantee, so we may assume that $0<c\leq1/2$. 

\paragraph{Action size bounds.} Given the $m$-dimensional \BiLinVI{} instance as above, we will construct a normal-form game instance as in \cref{sec:nf-game-con}, for an appropriate choice of $t$. For future, reference, we record the sizes of the action sets for each player:
\begin{align}
|\MAL|&=(t+1)h^4 2^t,&
|\MAR|&=h2^t+(2h+1)^t+2^t+1.
\end{align}
Consequently, there is an absolute constant $C_0$ such that the larger
action set has size at most
$2^{C_0t\log m}$. 
We will also use a smaller realization when $t=1$. In this case, take each
of $\MV$ and $\MW$ to contain one labeled copy of every row of the symmetric
Hadamard matrix $H_h$ and its negative, as in the discussion following
\cref{lem:scaled-hypercube}. It is straightforward to see that \cref{eq:rho-containment} continues to hold, and therefore that all lemmas from the previous subsection hold with this modified choice of $\MV,\MW$. Importantly, this modified choice ensures that that $|\MW|, |\MV| = O(h) = O(m)$ when $t=1$ and thus the number of actions for each player is $O(m)$.

\paragraph{Construction of the normal-form game.} We will next construct a normal-form game instance from the \BiLinVI{} instance, as in \cref{sec:nf-game-con}; note that this construction relies on a choice of $t \leq m$. To specify such $t$, we define
\begin{align}
\ep_N:=\ep(N),\qquad e_N:=4K\ep_N.
\end{align}
Choose sufficiently small absolute constants $c_1,c_2>0$ (to be specified below),
and a small absolute constant $c_3>0$ satisfying
$4K\sqrt{c_3/\ep_0}<1$. Define
\begin{align}
M(N):=\begin{cases}
\left\lfloor \frac{c_1(\log_2 N)^2}{\ep_N^2\log^2((\log_2 N)/\ep_N^2)}\right\rfloor,&\log((\log_2 N)/\ep_N^2)\leq c_2\log_2 N,\\
\left\lfloor c_3/\ep_N^2\right\rfloor,&\log((\log_2 N)/\ep_N^2)>c_2\log_2 N.
\end{cases}\nonumber
\end{align}
We claim that we can choose $N$ so that $m\leq M(N)\leq O(m)$. Indeed,
$M(N)\to\infty$ since $\ep_N=o(1)$. Moreover, slow decrease of $\ep_N$
implies that $\frac{\log_2(N+1)}{\ep_{N+1}^2}\big/\frac{\log_2 N}{\ep_N^2}=O(1)$, so each of the two expressions defining
$M(N)$ changes by at most a constant factor in one step. If consecutive
integers lie in different regimes, then $\log((\log_2 N)/\ep_N^2)=\Theta(\log_2 N)$ at both
integers, and both expressions are $\Theta(1/\ep_N^2)$ there. Thus
$M(N+1)=O(M(N))$ even at a transition between regimes. Taking the first
$N$ beyond a fixed sufficiently large threshold for which $M(N)\geq m$
proves the claim for all sufficiently large $m$. Henceforth we proceed with such a value of $N$. 

For this choice of $N$, we have $e_N\leq\ep_0$ and
$\log((\log_2 N)/\ep_N^2)=O(\log_2 N)$. Suppose first that $\log((\log_2 N)/\ep_N^2)\leq c_2\log_2 N$ (i.e., the ``small-$\ep$ regime''). Then we take
\begin{align}
m\leq\left\lfloor\frac{c_1(\log_2 N)^2}{\ep_N^2\log^2((\log_2 N)/\ep_N^2)}\right\rfloor=\Theta(m),\qquad
t:=\left\lceil e_N\sqrt{m/\ep_0}\right\rceil. \label{eq:general-parameters-final}
\end{align}
As long as $m$ is sufficiently large, we have $1\leq t\leq\sqrt m$ 
and
\begin{align}
t&\leq \frac{4K\sqrt{c_1}}{\sqrt{\ep_0}}\frac{\log_2 N}{\log((\log_2 N)/\ep_N^2)}+1,
&\log m&\leq2\log((\log_2 N)/\ep_N^2).
\end{align}
For $c_1$ and$c_2$ sufficiently small, the number of actions $2^{C_0 t \log m}$ is bounded above by $N$. Thus the reduction
maps every $m$-dimensional \BiLinVI{} instance to a game with at most $N$
actions per player in the ``small-$\ep$ regime''. 

Next we address the ``large-$\ep$'' regime $\log((\log_2 N)/\ep_N^2)>c_2\log_2 N$. Here, we have
\begin{align}
m\leq\left\lfloor\frac{c_3}{\ep_N^2}\right\rfloor=\Theta(m),\qquad t:=1.
\label{eq:small-t-parameters-final}
\end{align}
This is exactly the value of $t$ prescribed by
\cref{lem:derive-linvi-solution}. Moreover,
$m/N=O(1/(N\ep_N^2))=o(1)$ by the assumption
$\ep_N=\omega(1/\sqrt N)$. The $t=1$ realization above therefore also has
at most $O(m) \leq N$ actions per player as long as $m$ is sufficiently large. 

The payoffs of the normal-form game
$\MG$ (as constructed in \cref{sec:nf-game-con}) lie in $[-2K,2K]$. Normalize them to $[0,1]$ by setting
\begin{align}
\widetilde u^{\mathsf P}:=\frac{u^{\mathsf P}+2K}{4K},\qquad
\mathsf P\in\{\mathsf L,\mathsf R\}.
\end{align}
The assumed algorithm returns an $\ep_N$-Nash equilibrium of this normalized
game, which is an $e_N$-Nash equilibrium of $\MG$. Our choice
$t=\lceil e_N\sqrt{m/\ep_0}\rceil$ therefore allows
\cref{lem:derive-linvi-solution} to decode an $\eta$-approximate solution
of the original \BiLinVI{} instance.

\paragraph{Running time bound.} Finally, we bound the running time. Its base-two logarithm (by assumption in the theorem statement) is at most
\begin{align}
(\log_2 N)\left(\frac{\log_2 N}{\ep_N^2}\right)^{1-c}. \label{eq:runtime-final}
\end{align}
In the regime of \cref{eq:general-parameters-final}, we have
$m=\Theta(\frac{(\log_2 N)^2}{\ep_N^2\log^2((\log_2 N)/\ep_N^2)})$ and $\log m=\Theta(\log((\log_2 N)/\ep_N^2))$. Hence, for every
fixed constant $C>0$,
\begin{align}
\frac{(\log_2 N)\left(\frac{\log_2 N}{\ep_N^2}\right)^{1-c}}{m/(\log m)^C}
=O\left(\frac{(\log((\log_2 N)/\ep_N^2))^{C+2}}{((\log_2 N)/\ep_N^2)^c}\right)=o(1).
\end{align}
In the regime of \cref{eq:small-t-parameters-final}, we have
$(\log_2 N)/\ep_N^2=\Theta(m\log_2 N)$ and $\log m=\Theta(\log_2 N)$, and therefore
\begin{align}
(\log_2 N)\left(\frac{\log_2 N}{\ep_N^2}\right)^{1-c}=O\left(m^{1-c}(\log_2 N)^{2-c}\right)
=o\left(\frac{m}{(\log m)^C}\right)
\end{align}
for every fixed $C>0$. The search for $N$, construction of the explicit payoff matrices,
and decoding take $N^2\operatorname{poly}(m)$ time, whose base-two
logarithm is $O(\log_2 N+\log m)=o(m/(\log m)^C)$ in both regimes. We have
therefore obtained a $2^{\widetilde{o}(m)}$-time procedure for
every sufficiently large source dimension $m$, contradicting the hardness
of the original \BiLinVI{} instance in \cref{lem:bilinvi-hard}.
\end{proof}

\section{A tight lower bound for free games}
\label{sec:free-game}

In this section, we prove \cref{thm:freegame-main}.

\subsection{The free-game instance}
\label{sec:freegame-instance}
Fix a $D$-regular 2-CSP instance $\mathcal I = (\VL, \VR, E, \Sigma_{\mathsf{L}}, \Sigma_{\mathsf{R}}, (P_e)_{e \in E})$ from
\cref{lem:freegame-source-normalization} (we will refer to $\mathcal{I}$ as the \emph{source instance}). We write $m = |\VL| = |\VR|$. We next construct a free game instance, 
parameterized by a power of two $t$ satisfying $t \leq m$, and an integer $k$.  Given such $t,k$, define $h := 2^{\lceil \log(m/t) \rceil}$, so that $m \leq th < 2m$. We will ensure that $k$ satisfies
\begin{align}
4\log h&\leq k\leq h/4.
\label{eq:freegame-parameters}
\end{align}
 Apply \cref{lem:partitions} to the support graph of
$\mathcal I$ to obtain blocks
$I_1,\ldots,I_t \subset \VL$ of size at most $h$ and blocks
$J_1,\ldots,J_h \subset \VR$ of size at most $t$.  The lemma ensures that at most
$2D(D^2+3)$  edges belong to $I_s \times J_g$ for any $s \in [t], g \in [h]$.  Since each tuple $(i,j) \in \VL \times \VR$ has at most $D$ parallel constraint occurrences, every pair of blocks $I_s, J_g$ contains at
most
$
C_{\mathsf{fg}}:=2D^2(D^2+3)
$
constraint occurrences.

Let $\MV=\MV(h,k)$ be the multiset from
\cref{lem:scaled-hypercube}.  Thus
\begin{align}
|\MV|=h^4 2^k,\qquad
\rho[-1,1]^h\subseteq\operatorname{conv}(\MV),\qquad
\rho:=\frac14\sqrt{\frac{k}{h}}.
\label{eq:freegame-compressed-family}
\end{align}
It follows from the containment of $\rho[-1,1]^h$ in $\operatorname{conv}(\MV)$ that for every $z\in\BR^h$,
\begin{align}
\max_{v\in\MV}\langle v,z\rangle\geq\rho\|z\|_1.
\label{eq:freegame-support-bound}
\end{align}
Distinct indices of the multiset are treated as distinct answer symbols.
Finally, we define
$p:=\log|\MV|=k+4\log h.$
Then \cref{eq:freegame-parameters} implies $p\leq2k$ and
$2p\leq h$.

We next construct describe the construction of a family of $p$-wise independent hash functions which, roughly speaking, will ensure that the left prover only plays actions corresponding to ``valid'' assignments in the source 2-CSP instance $\mathcal{I}$.

Recalling that $h$ is a power of $2$ (so that the field $\F_h$ is well-defined), we view $\F_h$ as a vector space over
$\F_2$ and fix any nonzero linear map
$L:\F_h\to\F_2$ (for example, $L$ may return the first
coordinate in a fixed basis; or we may use the trace $\operatorname{Tr}_{\F_h/\F_2}$). For each 
$
\theta=(\theta_0,\ldots,\theta_{2p-1})
\in\F_h^{2p},
$
define
\begin{align}
Q_\theta(\zeta)&:=\sum_{r=0}^{2p-1}\theta_r\zeta^r,&
\beta_\theta(x)&:=
(-1)^{L(Q_\theta(x))}.
\label{eq:freegame-mask}
\end{align}
For a uniform $\theta \sim \mathrm{Unif}(\F_h^{2p})$, the signs
$(\beta_\theta(x))_{x\in\F_h}$ are $2p$-wise independent and unbiased (here we use that $2p \leq h$).

The following lemma bounds the correlation between any fixed $z \in \mathbb{R}^h$ and the output of some function $f_{\mathrm{sel}} : \mathbb{F}_h^{2p} \to \MV$ over a uniform choice of its input $\theta \in \mathbb{F}_h^{2p}$. We use $\odot$ to denote the entrywise product of two vectors in $\mathbb{R}^h$.

\begin{lemma}
\label{lem:freegame-adaptive-selector}
For every $z\in\BR^h$ and every map
$f_{\mathrm{sel}}:\F_h^{2p}\to\MV$,
\begin{align}
\E_{\theta\in\F_h^{2p}}
\left|
\left\langle\beta_\theta\odot f_{\mathrm{sel}}(\theta),z\right\rangle
\right|
\leq2\sqrt p\,\|z\|_2.
\label{eq:freegame-adaptive-selector}
\end{align}
\end{lemma}

\begin{proof}
Fix $v\in\MV$, and write
$w_x:=v_xz_x$ for $x \in [h]$ and
$Z_\theta:=\langle\beta_\theta\odot v,z\rangle
=\sum_{x\in[h]}\beta_\theta(x)w_x$.
We may compute
\begin{align*}
\E_\theta|Z_\theta|^{2p}
=\sum_{x_1,\ldots,x_{2p}\in[h]}
\left(\prod_{j=1}^{2p}w_{x_j}\right)
\E_\theta\left[\prod_{j=1}^{2p}\beta_\theta(x_j)\right].
\end{align*}
By $2p$-wise independence and unbiasedness, the expectation in each
summand is zero unless every coordinate occurs an even number of times, in
which case it is one.  Every surviving tuple can be paired so that the two
indices in each pair agree.  Summing over the $(2p-1)!!$ pairings therefore
gives
\begin{align*}
\E_\theta|Z_\theta|^{2p}
&\leq(2p-1)!!\left(\sum_{x\in[h]}w_x^2\right)^p
\leq(2p-1)!!\,\|z\|_2^{2p}
\leq(2p)^p\|z\|_2^{2p}.
\end{align*}
Here the second inequality uses $|v_x|=1$, and the last uses
$(2p-1)!!\leq(2p)^p$.
Consequently, using $|\MV|=2^p$,
\begin{align*}
\E_\theta\max_{v\in\MV}
\left|\left\langle\beta_\theta\odot v,z\right\rangle\right|
&\leq
\left(
\sum_{v\in\MV}\E_\theta
\left|\left\langle\beta_\theta\odot v,z\right\rangle\right|^{2p}
\right)^{1/(2p)}\leq |\MV|^{1/(2p)}\sqrt{2p}\,\|z\|_2
=2\sqrt p\,\|z\|_2.
\end{align*}
The left-hand side of \cref{eq:freegame-adaptive-selector} is bounded
above by the maximum in the preceding display.
\end{proof}

\paragraph{Construction of the hard free game instance.} We now specify the components of the free game
$\mathcal F(\mathcal I)=(X,Y,A,B,V)$.  
The question and answer sets are
\begin{align}
X&:=[t]\times\{-1,1\}^{\Sigma_{\mathsf L}}
       \times\F_h^{2p},&
Y&:=[h],
\label{eq:freegame-questions}\\
A&:=\MV,&
B&:=\Sigma_{\mathsf R}^{\,t}
       \times\Sigma_{\mathsf L}^{\,Dt}.
\label{eq:freegame-answers}
\end{align}
Each element of $A$ corresponds to an element of the multiset $\MV=\MV(h,k)$.  A  question for the left-player (i.e., $(X,A)$ player) is a tuple
$x=(s,\tau,\theta)$, where $s$ selects $I_s$, $\tau$ should be interpreted as a hash function
$\tau:\Sigma_{\mathsf L}\to\{-1,1\}$, and $\theta$  should be interpreted as a seed for the $(2p)$-wise independent family in \cref{eq:freegame-mask}.  A  question for the right-player (i.e., the $(Y,B)$-player) $y=g$ selects $J_g$.  An answer for the right player gives one
label $b_j\in\Sigma_{\mathsf R}$ for every $j\in J_g$ and a separate
claimed label $a_e\in\Sigma_{\mathsf L}$ for every constraint occurrence
$e=(i,j)$ incident to $J_g$. 

Next, fix an arbitrary injective map 
$\lambda_s:I_s\to[h]$ for each $s \in [t]$. 
Given questions $(s,\tau,\theta)\in X$ and $g\in Y$, a left answer
$v\in A$, and a right answer represented by $(b_j,a_e)_{j,e}$ as above, define
the raw score
\begin{align}
S\big((s,\tau,\theta),g,v,(b_j,a_e)_{j,e}\big)
:=\frac{1}{C_{\mathsf{fg}}}
\sum_{\substack{e=(i,j)\in E\cap(I_s\times J_g)\\
                 P_e(a_e,b_j)=1}}
\beta_\theta(\lambda_s(i)) \cdot v_{\lambda_s(i)}\cdot \tau(a_e),
\label{eq:freegame-raw-verifier}
\end{align}
and set
\begin{align}
V\big((s,\tau,\theta),g,v,(b_j,a_e)_{j,e}\big)
:=\frac{1+S\big((s,\tau,\theta),g,v,(b_j,a_e)_{j,e}\big)}{2}.
\label{eq:freegame-verifier}
\end{align}
There are at most $C_{\mathsf{fg}}$ summands in
\cref{eq:freegame-raw-verifier} (by definition of $C_{\mathsf{fg}}$), each in $\{-1,1\}$, so $S\in[-1,1]$ and
$V\in[0,1]$.  In particular, this is a valid free game whose verifier
entries are rationals of constant bit complexity.

To get some intuition for the expression $S((s,\tau,\theta), g,v,(b_j, a_e)_{j,e})$ in \cref{eq:freegame-raw-verifier}, at a high level the goal is to count the number of constraints $e = (i,j) \in E \cap (I_s \times J_g)$ inside the chosen block tuple $I_s \times J_g$ for which the right-player's label $b_j$ and their ``proposal'' for the left-player's label, $a_e$, satisfy the constraint $P_e(a_e, b_j) = 1$. However, of course the right-player could cheat by ``proposing'' the labels $a_e$ in some inconsistent way. To prevent the right player from cheating in this way, we need a few additional components: first, the left player action $v \in \MV$ should be viewed as proposing a ``hash'' for the label $a \in \Sigma_{\mathsf{L}}$ that they would like to assign to each element $i \in I_s$, which should line up with the actual hash $\tau(a_e)$ of the label $a_e$ chosen by the right-player (recall that $\tau$ is chosen as part of the question, and will be uniformly random); hence the product $v_{\lambda_s(i)} \cdot \tau(a_e)$. (Note that if we had $|\Sigma_{\mathsf{L}}| = 2$, then this hash would be unnecessary.) 

Finally, to prevent the left-player from cheating in their choice of $v$ (as $\mathrm{conv}(\MV)$ is a strict superset of the scaled hypercube $\rho[-1,1]^h$), we include an additional hash, $\beta_\theta(\lambda_s(i))$, which is a random sign that is $2p$-wise independent across the coordinates of $v$. Its role in preventing the left-player from cheating may be observed from \cref{lem:freegame-adaptive-selector}: the mapping $f_{\mathrm{sel}} : \mathbb{F}_h^{2p} \to \MV$ should be understood as the left-player's choice of action (which may depend on $\theta$), and \cref{lem:freegame-adaptive-selector} shows that, in expectation over the seed $\theta$, when incorporating the hash $\beta_\theta$, $f_{\mathrm{sel}}(\theta)$ ``behaves like'' a vector of $\ell_2$ norm at most $O(\sqrt{p})$, which is exactly the $\ell_2$ norm of vectors in $\rho\{-1,1\}^h$ (recall that $\rho = \Theta(\sqrt{k/h})$ and $p \leq O(k)$).

\subsection{Completeness and soundness}

We next prove completeness of the above reduction. 

\begin{lemma}[Completeness]
\label{lem:freegame-completeness}
If $\operatorname{val}(\mathcal I)=1$, then
\begin{align}
\omega(\mathcal F(\mathcal I))
\geq \frac12+\frac{D}{16C_{\mathsf{fg}}}
\sqrt{\frac{k}{h}}.
\label{eq:freegame-completeness}
\end{align}
\end{lemma}

\begin{proof}
Fix a satisfying labeling
$(a_i^\star)_{i\in\VL},(b_j^\star)_{j\in\VR}$ for the source instance $\mathcal I$.  On question $g$, the
right prover reports $b_j=b_j^\star$ for every $j\in J_g$ and
$a_e=a_i^\star$ for every incident occurrence $e=(i,j)$.  For each left
question $(s,\tau,\theta)$, form $z\in\BR^h$ by setting
\begin{align}
z_{\lambda_s(i)}
:=D \cdot \beta_\theta(\lambda_s(i))\cdot \tau(a_i^\star)
\qquad(i\in I_s).
\label{eq:freegame-completeness-z}
\end{align}
 By
\cref{eq:freegame-support-bound}, there is a deterministic answer
$v\in\MV$ satisfying
\begin{align*}
\langle v,z\rangle
\geq\rho\|z\|_1=\rho D|I_s|.
\end{align*}
Choose one such maximizing answer for every left-player  question.

All constraints pass under these responses.  Averaging the raw score $S$ (defined in \cref{eq:freegame-raw-verifier})
over $g\in[h]$ and $s\in[t]$ gives
\begin{align*}
\E_{s,\tau,\theta,g}[S]
&=\frac{1}{C_{\mathsf{fg}}th}
\sum_{s=1}^t\E_{\tau,\theta}\langle v,z\rangle
\geq\frac{\rho Dm}{C_{\mathsf{fg}}th}
\geq\frac{\rho D}{2C_{\mathsf{fg}}},
\end{align*}
where the first equality uses that the source instance $\mathcal{I}$ is $D$-regular (so that each $j \in \VR$ is incident to $D$ constraints), and the final inequality uses $th<2m$.  The verifier value is one half
plus one half of this raw score.  Substituting
$\rho=\frac14\sqrt{k/h}$ proves
\cref{eq:freegame-completeness}.
\end{proof}

Next we establish soundness of our reduction. 

\begin{lemma}[Soundness]
\label{lem:freegame-soundness}
If $\operatorname{val}(\mathcal I)\leq\sigma_\star$, then
\begin{align}
\omega(\mathcal F(\mathcal I))
\leq \frac12+\frac{D}{C_{\mathsf{fg}}}
\sqrt{\frac{p\sigma_\star}{h}}.
\label{eq:freegame-soundness}
\end{align}
\end{lemma}

\begin{proof}
Fix arbitrary deterministic response functions for the two provers.  Since
the $J_g$ are disjoint, the right-player responses associate a label
$b_j\in\Sigma_{\mathsf{R}}$ to every $j\in\VR$, and a label
$a_e \in \Sigma_{\mathsf{L}}$ for every constraint $e \in E$. For $i\in\VL$ and
$a\in\Sigma_{\mathsf L}$, define, for $a \in \Sigma_{\mathsf{L}}$,
\begin{align}
d_i(a):=
\#\{e=(i,j)\in E:a_e=a\text{ and }P_e(a,b_j)=1\}.
\label{eq:freegame-di}
\end{align}
Choosing at each $i$ a label $a$ maximizing $d_i(a)$ gives a global labeling that satisfies at least
$\sum_i\max_a d_i(a)$ constraints of the source instance $\mathcal{I}$. Since we have assumed $\mathrm{val}(\mathcal I)\leq\sigma_\star$, it holds that 
\begin{align}
\sum_{i\in\VL}\max_{a\in\Sigma_{\mathsf L}}d_i(a)
\leq\sigma_\star Dm.
\label{eq:freegame-source-soundness-use}
\end{align}
For $s\in[t]$ and a label hash $\tau : \Sigma_{\mathsf{L}} \to \{-1,1\}$, define
$z_{s,\tau}\in\BR^h$ by
\begin{align}
(z_{s,\tau})_{\lambda_s(i)}
:=\sum_{a\in\Sigma_{\mathsf L}}d_i(a)\tau(a)
\qquad(i\in I_s).
\label{eq:freegame-soundness-z}
\end{align}
Fix any vector
$F(s,\tau,\theta)\in\MV$, corresponding to the left-player's answer in response to the query $(s,\tau,\theta)$.  Summing
\cref{eq:freegame-raw-verifier} over all left-player blocks shows that its
expected raw score is exactly
\begin{align}
\frac{1}{C_{\mathsf{fg}}th}
\sum_{s=1}^t\E_{\tau,\theta}
\left\langle
\beta_\theta\odot F(s,\tau,\theta),z_{s,\tau}
\right\rangle.
\label{eq:freegame-soundness-raw}
\end{align}
Applying \cref{lem:freegame-adaptive-selector} for each fixed value of $(s,\tau)$ (so in particular $f_{\mathrm{sel}}(\theta) = F(s,\tau,\theta)$) gives that 
\cref{eq:freegame-soundness-raw} is at most
\begin{align}
\frac{2\sqrt p}{C_{\mathsf{fg}}th}
\sum_{s=1}^t\E_\tau\|z_{s,\tau}\|_2.
\label{eq:freegame-soundness-mask-bound}
\end{align}
The signs $(\tau(a))_{a\in\Sigma_{\mathsf L}}$ are independent and
unbiased.  Moreover, $\sum_a d_i(a)\leq D$, so
\begin{align}
\E_\tau\left(\sum_a d_i(a)\tau(a)\right)^2
=\sum_a d_i(a)^2
\leq D\max_a d_i(a).
\label{eq:freegame-hash-moment}
\end{align}
Using Cauchy--Schwarz, Jensen's inequality,
\cref{eq:freegame-source-soundness-use}, and
\cref{eq:freegame-hash-moment}, we obtain
\begin{align*}
\sum_{s=1}^t\E_\tau\|z_{s,\tau}\|_2
&\leq
\sqrt{t\sum_{s=1}^t\E_\tau\|z_{s,\tau}\|_2^2}\leq
\sqrt{tD\sum_{i\in\VL}\max_a d_i(a)}
\leq D\sqrt{t\sigma_\star m}.
\end{align*}
Since $m\leq th$, the expected raw score in
\cref{eq:freegame-soundness-mask-bound} is at most
\begin{align*}
\frac{2D}{C_{\mathsf{fg}}}\sqrt{\frac{p\sigma_\star}{h}}.
\end{align*}
Taking one half of this raw bias proves
\cref{eq:freegame-soundness}.
\end{proof}

For later use, define the two thresholds
\begin{align}
v_{\mathrm{yes}}
&:=\frac12+\frac{D}{16C_{\mathsf{fg}}}\sqrt{\frac{k}{h}},&
v_{\mathrm{no}}
&:=\frac12+\frac{D}{C_{\mathsf{fg}}}
\sqrt{\frac{p\sigma_\star}{h}}.
\label{eq:freegame-thresholds}
\end{align}
Since $p\leq2k$ and $\sigma_\star=2^{-24}$, we have 
\begin{align}
v_{\mathrm{yes}}-v_{\mathrm{no}}
\geq\frac{D}{32C_{\mathsf{fg}}}\sqrt{\frac{k}{h}}.
\label{eq:freegame-value-gap}
\end{align}

\subsection{Proof of the free-game lower bound}

We are now ready to prove \cref{thm:freegame-main}.
\begin{proof}[Proof of \cref{thm:freegame-main}]
Fix a $D$-regular 2-CSP instance $\mathcal{I}$ from \cref{lem:freegame-source-normalization} where each side has size $m$. Since $N \mapsto \ep(N)$ is slowly decreasing, after increasing $m$ by at most a constant factor (and padding the resulting CSP instance), we can find $N \in \BN$ so that, for $\ep := \ep(N)$, we have $m = t \cdot h$ for 
\begin{align}
t:=2^{\lfloor\log(c_1\log N)\rfloor}, \qquad 
h:=2^{\left\lfloor\log\left(
\frac{c_1\log N}
{\ep^2\left\lceil\log((\log N)/\ep^2)\right\rceil}
\right)\right\rfloor}.\nonumber
\end{align}
The value of $h$ prescribed above (given the values of $t,m$) is exactly $2^{\lceil \log(m/t) \rceil}$, as required in the construction in \cref{sec:freegame-instance}. Next, set $\Lambda := (128 C_{\mathsf{fg}}/D)^2$ and $k = \lceil \Lambda \ep^2 h \rceil$. Consider the free game instance $\mathcal{F}(\mathcal{I}) = (X,Y,A,B,V)$ constructed in \cref{sec:freegame-instance} with the above choices of $t, k$ and $h$. 
From \cref{eq:freegame-questions,eq:freegame-answers}, we have
\begin{align}
|X|&=t2^{q_{\mathsf L}}h^{2p},&
|Y|&=h,&
|A|&=h^4 2^k,&
|B|&=q_{\mathsf R}^tq_{\mathsf L}^{Dt}.
\label{eq:freegame-cardinalities}
\end{align}
Consequently, the size $n_0 := |\mathcal{F}(\mathcal{I})|$ of the constructed free game is 
$
n_0=t2^{q_{\mathsf L}}h^{2p+5}2^k
q_{\mathsf R}^tq_{\mathsf L}^{Dt}.
$
Using $p=k+4\log h$ and $4\log h\leq k$, we obtain
$
\log n_0\leq C_{\mathsf{size}}(t+k\log h)
$
for an absolute constant $C_{\mathsf{size}}\geq1$.

Choose a sufficiently small
absolute constant $c_1>0$ so that
$C_{\mathsf{size}}(\Lambda+2)c_1\leq1/2$.  We will prove the theorem for a
sufficiently small constant $\alpha^\star>0$ satisfying
$9(\alpha^\star)^2\leq\Lambda c_1/8$.

Recalling that $\ep = \ep(N) \geq 2^{-\alpha^\star\sqrt{\log N}}$, we have
\begin{align}
\left\lceil\log\left(\frac{\log N}{\ep^2}\right)\right\rceil
&\leq\log\log N+2\alpha^\star\sqrt{\log N}+1\leq 3\alpha^\star\sqrt{\log N}\label{eq:alphastar-property}
\end{align}
for all sufficiently large $N$.  

It remains to check that $k$ satisfies \cref{eq:freegame-parameters}.  First,
\begin{align*}
k
&\geq\Lambda\ep^2h
>\frac{\Lambda c_1\log N}
{2\left\lceil\log((\log N)/\ep^2)\right\rceil}\geq4\left\lceil\log\left(\frac{\log N}{\ep^2}\right)\right\rceil
\geq4\log h,
\end{align*}
where the penultimate inequality follows from the choice of
$\alpha^\star$ and \cref{eq:alphastar-property}, and the last follows from the definition of $h$.  On the
other hand,
\begin{align*}
\frac{k}{h}\leq\Lambda\ep^2+\frac1h=o(1),
\end{align*}
because $\ep=o(1)$ and $h\to\infty$.  Hence $k\leq h/4$ for all
sufficiently large $N$, as required by
\cref{eq:freegame-parameters}.

Let $n_0$ be the size of the resulting game. We may bound
\begin{align*}
\log n_0
&\leq C_{\mathsf{size}}(t+k\log h)
\leq C_{\mathsf{size}}\left(
(\Lambda+1)c_1\log N
+\left\lceil\log\left(\frac{\log N}{\ep^2}\right)\right\rceil
\right)
\leq\log N
\end{align*}
for all sufficiently large $N$; in the last step we used
$\left\lceil\log((\log N)/\ep^2)\right\rceil=o(\log N)$ and the choice of
$c_1$. By padding, we may ensure that the size is exactly $N$.
By \cref{lem:freegame-completeness,lem:freegame-soundness} and \cref{eq:freegame-value-gap}, an algorithm which obtains an $\ep$-additive approximate estimate of the value value of the free game $\mathcal{F}(\mathcal{I})$ distinguishes between the cases $\mathrm{val}(\mathcal{I}) \leq \sigma_\star$ and $\mathrm{val}(\mathcal{I}) = 1$; here we use in particular that (by \cref{eq:freegame-value-gap})
\begin{align}
v_{\mathrm{yes}}-v_{\mathrm{no}}
\geq\frac{D}{32C_{\mathsf{fg}}}\sqrt{\frac{k}{h}}
\geq\frac{D\sqrt\Lambda}{32C_{\mathsf{fg}}}\ep
=4\ep.
\label{eq:freegame-calibrated-gap}
\end{align}

Finally, the parameter choices give
\begin{align}
\frac{c_1^2(\log N)^2}
{4\ep^2\left\lceil\log((\log N)/\ep^2)\right\rceil}
<m\leq
\frac{c_1^2(\log N)^2}
{\ep^2\left\lceil\log((\log N)/\ep^2)\right\rceil}
\label{eq:freegame-source-target-relation}
\end{align}
for all sufficiently large $N$.  The base-two logarithm of the alleged
algorithm's running time (i.e., from the statement of \cref{thm:freegame-main}) is at most
\begin{align*}
\log N\left(\frac{\log N}{\ep^2}\right)^{1-c}
\leq\frac{4}{c_1^2}m
\frac{\left\lceil\log((\log N)/\ep^2)\right\rceil}
{((\log N)/\ep^2)^c}\leq m^{1-c/8}.
\end{align*}
The last inequality uses
$\left\lceil\log((\log N)/\ep^2)\right\rceil
\leq((\log N)/\ep^2)^{c/4}$ and
$m\leq((\log N)/\ep^2)^2$, and holds for all sufficiently large $N$.
Summarizing, we have shown that the resulting algorithm distinguishes between $\mathrm{val}(\mathcal{I}) \leq \sigma_\star$ and $\mathrm{val}(\mathcal{I}) = 1$ in time $2^{O(m^{1-c/16})}$, which contradicts \cref{lem:freegame-source-normalization} and thus proves the theorem. 
\end{proof}

\section{Detailed AI Use Statement}
\label{sec:ai-use}

I first asked GPT-5.6-sol ultra to find a proof that the algorithm of \cite{LMM03} which computes $\ep$-Nash equilibrium in time $N^{O(\log(N)/\ep^2)}$ is optimal, or to find a faster algorithm. I used a variant of OpenAI's cycle double cover prompt\footnote{\url{https://cdn.openai.com/pdf/04d1d1e4-bc75-476a-97cf-49055cd98d31/cdc_prompt.pdf}}. After a few steps of encouragement (e.g., ``keep working until you find a solution, you may use additional subagents.''), it eventually produced a proof that there is no algorithm running in time $2^{1/\ep^{4/3}}$, for some $\ep$ which scaled as an inverse polynomial in the number of actions $N$. I then told it to try to improve this to $2^{1/\ep^2}$, and it did so (again in a parameter regime where $\ep \asymp 1/\mathrm{poly}(N)$). I next instructed it to use ideas in ``birthday repetition'' to prove a lower bound of $N^{\log(N)/\ep^{2-c}}$ with $\ep \asymp 1/\log(N)$ (which it succeeded in doing), and then I instructed it to generalize the proof to work for arbitrary $\ep(N)$ as in the statement of \cref{thm:main-ne}.
Finally, I asked it to use the ideas which it obtained for showing hardness of approximate Nash equilibrium to show the tight hardness result for free games in \cref{thm:freegame-main}, which it succeeded in doing.

While my prompts required some knowledge of the material, any such knowledge required felt quite shallow and unlikely to constitute any meaningful obstacle for future AI models. On the other hand, the initial proof produced by GPT was incredibly difficult to read, and it required substantial manual editing to get it to a more reasonable form. Doing so required me to understand the proofs, though this necessity may not persist very long, i.e., one might expect that AI's ability at writing polished and easily understandable proofs and technical overviews (on its own) may significantly improve in the near term.

In two words, my role with respect to AI would best be described as (a) cheerleader; (b) cleaner. %

\section*{Acknowledgments.}
I thank Aviad Rubinstein for a useful discussion that spurred me to ask GPT about the topics in this paper. 

\newpage

\appendix

\section{Related work}

\paragraph{Algorithms and hardness for Nash equilibria.}
Computing a Nash equilibrium to inverse-polynomial accuracy is well-known to be \PPAD-complete, including for
bimatrix games~\cite{DGP09,CDT09}. %
For constant approximation parameter $\ep$, Rubinstein proved
\PPAD-completeness of $\ep$-Nash computation in degree-three polymatrix games with two actions per
player~\cite{Rub15}. %

\paragraph{Free games.}
Free games are a version of two-prover games where players receive independent questions, and
can encode dense bipartite $2$-CSPs.  Aaronson,
Impagliazzo, and Moshkovitz gave a combinatorial deterministic
$n^{O((\log n)/\ep^2)}$-time approximation algorithm based on subsampling.
Brand\~ao and Harrow obtained the same asymptotic bound through a
linear-programming relaxation, and Bernasconi et al.\ recently recovered it
for any fixed number of provers using a general polynomial-cover
framework~\cite{AIM14,BH13,BCCF26}.  Aaronson et al.\ also showed, under ETH, an
$n^{\widetilde\Omega((\log n)/\ep)}$ lower bound for
$1/n\leq\ep\leq\Delta$, for an absolute constant $\Delta>0$, and explicitly
asked whether the linear dependence on $1/\ep$ could be improved to match the
quadratic dependence in the upper bound~\cite{AIM14}.  Manurangsi and
Raghavendra \cite{MR17} subsequently proved a general birthday-repetition theorem and
derived stronger multiplicative inapproximability for free games and dense
CSPs; these results concern the regime where one aims for a multiplicative approximation of the value rather than the
additive approximation considered here.

\paragraph{Fine-grained lower bounds.}
For the optimization problem of finding a constant-approximate Nash
equilibrium with maximum social welfare, Braverman, Ko, and Weinstein showed
that an $N^{o(\log N)}$-time algorithm would refute ETH~\cite{BKW15}; their
reduction proceeds through the free-game framework of Aaronson,
Impagliazzo, and Moshkovitz~\cite{AIM14}. Babichenko, Papadimitriou, and
Rubinstein subsequently formulated the \PCP-for-\PPAD conjecture and used
the ``birthday repetition'' technique of \cite{AIM14} to obtain quasipolynomial hardness for the search problem of computing \emph{any} approximate Nash equilibrium in normal-form games, under \PCP-for-\PPAD~\cite{BPR15}. This result matches the dependence on
$N$ in the upper bound of \cite{LMM03} when $\ep$ is fixed.  Rubinstein then showed that computing an $\ep$-Nash equilibrium for some
absolute constant $\ep>0$ requires time
$N^{\log^{1-o(1)}N}$~\cite{Rub16} \emph{without} relying on \PCP-for-\PPAD. 

Despite the above progress, the dependence on
a vanishing accuracy parameter $\ep$ remained unresolved prior to the present work. In a related optimization setting,
Bernasconi et al.\ highlight the analogous gap between the
$\ep^{-2}$ dependence in the upper bound and the $\ep^{-1}$ dependence in
the lower bound, for a range of problems which can be phrased as optimization of low-degree polynomials over convex sets \cite{BCCF26}. 

\paragraph{Other computational models.}
Harrow, Natarajan, and Wu proved unconditional nearly logarithmic lower
bounds on the sum-of-squares degree needed to optimize over Nash equilibria
at constant accuracy~\cite{HNW16}. A number of papers have also studied the complexity of computing approximate Nash equilibria in other computational models. For example, Babichenko and
Rubinstein obtained a polynomial lower bound for the communication complexity of constant-approximate Nash
equilibrium in two-player games~\cite{BR17}, which G\"o\"os and Rubinstein
strengthened to $N^{2-o(1)}$~\cite{GR18}. %

\section{Proofs of the preliminary and technical lemmas}

We first justify the \BiLinVI{} hardness used throughout the reduction.

\begin{proof}[Proof of \cref{lem:bilinvi-hard}]
Let $\ep_*,\delta_*>0$ be the constants from
\cref{conj:pcp-ppad}, and consider one of the hard bipartite polymatrix
games. Write $\VL$ and $\VR$ for the two sides, each of size $m$, and
identify each player's two actions with $\{-1,1\}$. For a mixed-strategy
profile, let $\bx,\by\in[-1,1]^m$ denote the players' expected actions on
the two sides.

For each $i\in\VL$, let $d_i^{\mathsf L}(\by)$ be the expected payoff of
action $1$ minus that of action $-1$. Since the game is polymatrix and
bipartite, there are a vector $\boldsymbol{\alpha}^{\mathsf L}$ and a
matrix $M^{\mathsf L}$ such that
\begin{align}
d^{\mathsf L}(\by)=M^{\mathsf L}\by+
\boldsymbol{\alpha}^{\mathsf L}. \label{eq:polymatrix-left-difference}
\end{align}
Indeed, the payoff difference contributed by an edge $(i,j)$ is an affine
function of $\by_j$. The two values of this function lie in $[-1,1]$, so
its constant and linear coefficients have magnitude at most $1$. It follows
from the degree-three promise that the entries of $M^{\mathsf L}$ have
magnitude at most $1$, the entries of
$\boldsymbol{\alpha}^{\mathsf L}$ have magnitude at most $3$, and every
row and column of $M^{\mathsf L}$ has at most three nonzero entries. Define
\begin{align}
\AL:=\frac13M^{\mathsf L},\qquad
\aL:=\frac13\boldsymbol{\alpha}^{\mathsf L}.
\end{align}
Define $M^{\mathsf R}$ and
$\boldsymbol{\alpha}^{\mathsf R}$ analogously from the action-$1$ versus
action-$-1$ payoff differences of the players in $\VR$, and set
\begin{align}
\AR:=\frac13M^{\mathsf R},\qquad
\aR:=\frac13\boldsymbol{\alpha}^{\mathsf R}.
\end{align}
Thus $(\AL,\AR,\aL,\aR)$ is a valid \BiLinVI{} instance, and each of
$\AL,\AR$ has at most three nonzero entries in every row and column.

It remains to relate the two notions of regret. Let $r_i$ denote the
ordinary unilateral regret of a player in the initial polymatrix game. For $i\in\VL$, linearity in
the player's own centered strategy gives
\begin{align}
r_i=\max_{z\in[-1,1]}
\frac12(z-\bx_i)d_i^{\mathsf L}(\by).
\end{align}
The maximization over the cube separates coordinatewise. Consequently,
using \cref{eq:polymatrix-left-difference},
\begin{align}
\max_{\bx'\in[-1,1]^m}\frac1{2m}
\left\langle\bx'-\bx,\AL\by+\aL\right\rangle
&=\frac1{3m}\sum_{i\in\VL}r_i. \label{eq:bilinvi-source-left}
\end{align}
The symmetric calculation gives
\begin{align}
\max_{\by'\in[-1,1]^m}\frac1{2m}
\left\langle\by'-\by,\AR\bx+\aR\right\rangle
&=\frac1{3m}\sum_{j\in\VR}r_j. \label{eq:bilinvi-source-right}
\end{align}

Set $\eta:=\ep_*\delta_*/3$. If $(\bx,\by)$ is an
$\eta$-approximate \BiLinVI{} solution, then
\cref{eq:bilinvi-source-left,eq:bilinvi-source-right} imply that the total
regret of all $2m$ source players is at most
$6m\eta=2m\ep_*\delta_*$. Hence fewer than $2m\delta_*$ players can have
regret greater than $\ep_*$. Mapping each centered coordinate $z$ to the
binary mixed strategy that plays action $1$ with probability $(1+z)/2$
therefore yields an $(\ep_*,\delta_*)$-weak Nash equilibrium. The
construction and this decoding are polynomial-time and preserve the
parameter $m$. Thus any algorithm violating the asserted lower bound for
the stated \BiLinVI{} problem would contradict \cref{conj:pcp-ppad}.
\end{proof}

We next prove the two combinatorial lemmas used in \cref{sec:nf-game-con}.

\begin{proof}[Proof of \cref{lem:scaled-hypercube}]
Write
\begin{align}
\lambda:=\frac14\sqrt{\frac{k}{r}}.
\end{align}
When $r=1$, necessarily $k=1$, and we take the two labeled signs. Assume
henceforth that $r\geq2$. Partition $[r]$ into $k$ nonempty groups
$G_1,\ldots,G_k$, each of size at most
$g:=\lceil r/k\rceil$. Since $r$ is a power of two, let
$\mathbb F_r$ be the field of order $r$, and enumerate its elements as
$\alpha_1,\ldots,\alpha_r$. View $\mathbb F_r$ as a vector space over
$\mathbb F_2$, fix any nonzero linear map
$L:\mathbb F_r\to\mathbb F_2$ (for example, $L$ may return the first
coordinate in a fixed basis), and, for
$c=(c_0,c_1,c_2,c_3)\in\mathbb F_r^4$, define
\begin{align}
\xi_i^c:=(-1)^{L(c_0+c_1\alpha_i+c_2\alpha_i^2+
c_3\alpha_i^3)},\qquad i\in[r].
\end{align}
For $s\in\{-1,1\}^k$, define $v^{c,s}\in\{-1,1\}^r$ by
\begin{align}
v_i^{c,s}:=\xi_i^c s_\ell\qquad\text{for }i\in G_\ell,
\end{align}
and let
\begin{align}
\MV(r,k):=\{v^{c,s}:c\in\mathbb F_r^4,
s\in\{-1,1\}^k\}
\end{align}
as a labeled multiset. Its size is exactly $r^4 2^k$.

Fix $z\in\mathbb R^r$ and, for a uniformly random $c\in\mathbb F_r^4$,
write
\begin{align}
S_\ell:=\sum_{i\in G_\ell}\xi_i^c z_i.
\end{align}
The signs $(\xi_i^c)_{i\in[r]}$ are four-wise independent: for any at
most four distinct field points, the corresponding evaluation map for a
degree-three polynomial is surjective (equivalently, one may interpolate
arbitrary values at those points). Thus the corresponding polynomial
values are independent and uniform in $\mathbb F_r$. Applying $L$ to each
value preserves independence, and each resulting bit is unbiased because
$L$ is a nonzero linear map. Therefore
\begin{align}
\E_c[S_\ell^2]&=\sum_{i\in G_\ell}z_i^2,&
\E_c[S_\ell^4]&\leq3\left(\sum_{i\in G_\ell}z_i^2\right)^2.
\end{align}
If $z$ vanishes on $G_\ell$, the following lower bound is immediate;
otherwise, $\E_c S_\ell^2>0$ by the preceding display. Applying
H\"{o}lder's inequality to $|S_\ell|^{2/3}$ and $|S_\ell|^{4/3}$ with
conjugate exponents $3/2$ and $3$ gives
\begin{align}
\E_c S_\ell^2
&=\E_c\left[|S_\ell|^{2/3}|S_\ell|^{4/3}\right]
\leq (\E_c|S_\ell|)^{2/3}(\E_c S_\ell^4)^{1/3}.
\nonumber
\end{align}
In particular, $\E_c S_\ell^4>0$, so raising this inequality to the
power $3/2$ and rearranging yields
\begin{align}
\E_c|S_\ell|
&\geq\frac{(\E_c S_\ell^2)^{3/2}}{(\E_c S_\ell^4)^{1/2}}
\geq\frac1{\sqrt3}\left(\sum_{i\in G_\ell}z_i^2\right)^{1/2}.
\end{align}
Consequently, some choice of $c$ satisfies
\begin{align}
\max_{s\in\{-1,1\}^k}\langle v^{c,s},z\rangle
&=\sum_{\ell=1}^k|S_\ell|\nonumber\\
&\geq\frac1{\sqrt3}\sum_{\ell=1}^k
\left(\sum_{i\in G_\ell}z_i^2\right)^{1/2}\nonumber\\
&\geq\frac1{\sqrt{3g}}\|z\|_1
\geq\lambda\|z\|_1. \label{eq:scaled-cube-support}
\end{align}
For the last inequality, we used
$g=\lceil r/k\rceil\leq2r/k$ and $1/\sqrt6\geq1/4$.
The right-hand side of \cref{eq:scaled-cube-support} is the support function
of $\lambda[-1,1]^r$. The support-function characterization of convex
containment therefore gives
$\lambda[-1,1]^r\subseteq\operatorname{conv}(\MV(r,k))$, which is the
claimed inclusion.

Finally, fix canonical encodings of the elements of $\mathbb F_r$ and of
the sign vectors. A label in $[r^4 2^k]$ then specifies the four field
elements $c$ and the sign vector $s$. A representation of $\mathbb F_r$ can be found
deterministically in $\operatorname{poly}(r)$ time by enumerating and
testing degree-$\log_2r$ polynomials over $\mathbb F_2$. Given this
representation, all polynomial evaluations, applications of $L$, and
coordinates of $v^{c,s}$ can be computed in $\operatorname{poly}(r)$ time.
\end{proof}

\begin{proof}[Proof of \cref{lem:partitions}]
Divide $\VL$ arbitrarily into $t$ labeled, possibly empty blocks whose
sizes differ by at most one. Their sizes are at most
$\lceil m/t\rceil\leq h$.

To construct the right blocks, form a conflict graph $H$ on $\VR$: two
distinct vertices are adjacent in $H$ if they both have a neighbor in the
same left block $\VL\^i$. A right vertex has neighbors in at most $D$ left
blocks, and each left block is incident to at most $Dh$ right vertices.
Thus
\begin{align}
\Delta(H)\leq D(Dh-1)\leq D^2h. \label{eq:conflict-degree}
\end{align}
Set $q:=2(D^2+3)h$. By the algorithmic Hajnal--Szemer\'edi theorem, $H$
has an equitable proper $q$-coloring computable in polynomial
time~\cite{KK08}: apply the theorem with degree parameter $q-1$, noting
from \cref{eq:conflict-degree} that $\Delta(H)\leq q-1$.
(If $q>m$, assign distinct colors to the vertices and
leave the remaining color classes empty.) Hence every color class has size
$\lfloor m/q\rfloor$ or $\lceil m/q\rceil$.

Group the $q$ color classes into $h$ groups of exactly $2(D^2+3)$ classes
each, distributing the larger color classes as evenly as possible among
the groups, and let $\VR\^1,\ldots,\VR\^h$ be their unions. The resulting
right-block sizes differ by at most one, and hence are at most
\begin{align}
\left\lceil\frac mh\right\rceil\leq t.
\end{align}
Fix a left block $\VL\^i$. Since each color class is independent in $H$,
it contains at most one right vertex having a neighbor in $\VL\^i$. Such a
vertex contributes at most $D$ edges from $\VL\^i$. Each right block is a
union of $2(D^2+3)$ color classes, so
\begin{align}
|E\cap(\VL\^i\times\VR\^j)|\leq2D(D^2+3)
\end{align}
for every $i,j$. All steps of the construction are deterministic and
polynomial-time.
\end{proof}

Finally, we prove \cref{lem:freegame-source-normalization}, which is an immediate consequence of the Raz-Moshkovitz constant-soundness PCP theorem~\cite{MR08}.
\begin{proof}[Proof of \cref{lem:freegame-source-normalization}]
Fix a $3$-CNF formula with $r$ variables and $O(r)$ constraints.  Apply
\cite[Theorem~30]{AIM14} with soundness $\sigma_\star$.  The resulting
bipartite $2$-CSP instance $\mathcal{I}$ satisfies $\mathrm{val}(\mathcal{I}) \leq \sigma_\star$ if the initial formula was unsatisfiable and $\mathrm{val}(\mathcal{I}) = 1$ otherwise.  Furthermore, the instance $\mathcal{I}$ has  alphabet size $2^{\operatorname{poly}(1/\sigma_\star)}$,
and $N=r^{1+o(1)}\operatorname{poly}(1/\sigma_\star)=r^{1+o(1)}$ total
variables.  Moreover, every variable participates in exactly
$d=\operatorname{poly}(1/\sigma_\star)$ constraints.  Since
$\sigma_\star$ is fixed, both the alphabet size and $d$ are absolute
constants. Let the number of variables on each side of the resulting 2-CSP instance be denoted by $m$. If an algorithm could distinguish between the cases $\mathrm{val}(\mathcal{I}) \leq \sigma_\star$ and $\mathrm{val}(\mathcal{I}) = 1$ for any such instance $\mathcal{I}$, it follows that we have an algorithm for $3$-CNF running in time
$2^{O(m^{1-\eta})}$.  Since $m=r^{1+o(1)}$,
composing the PCP reduction with the alleged distinguisher would therefore
decide the $3$-CNF formula in time
$2^{O(r^{1-\eta/2})}=2^{o(r)}$, which contradicts
\cref{conj:eth-freegame}.
\end{proof}

\end{document}